\def\year{2020}\relax
\documentclass[letterpaper]{article} % DO NOT CHANGE THIS
\usepackage{aaai20}  % DO NOT CHANGE THIS
\usepackage{times}  % DO NOT CHANGE THIS
\usepackage{helvet} % DO NOT CHANGE THIS
\usepackage{courier}  % DO NOT CHANGE THIS
\usepackage[hyphens]{url}  % DO NOT CHANGE THIS
\usepackage{graphicx} % DO NOT CHANGE THIS
\usepackage{graphicx}  % DO NOT CHANGE THIS
\usepackage{algorithmic}
\usepackage{algorithm}
\usepackage{bm}
\usepackage{amsmath}
\usepackage{amsthm}
\usepackage{amsfonts}
\usepackage{multirow}
\usepackage{float}
\usepackage{mathrsfs}
\theoremstyle{plain}

\title{Optimal Procurement Auction for Cooperative Production of Virtual Products: Vickrey-Clarke-Groves Meet Cr{\'e}mer-McLean}
\author{
Mingshu Cong
\thanks{
    This work was supported by LogiOcean Technologies, Ltd.
}\\
The FinTech and Blockchain Laboratory\\
The University of Hong Kong\\
LogiOcean Technologies, Ltd. \\
miranda.cong@logiocean.com
\And
\Large \textbf{
        Xi Weng}\\
Guanghua School of Management\\
Peking University\\
wengxi125@gsm.pku.edu.cn
\AND
\Large \textbf{
    Han Yu} \\
School of Computer Science \& Engineering\\
Nanyang Technological University\\
han.yu@ntu.edu.sg
\And
\Large \textbf{
    Jiabao Qu} \\
LogiOcean Technologies, Ltd.\\
qujiabao@logiocean.com
\And
\Large \textbf{
    Siu Ming Yiu} \\
The FinTech and Blockchain Laboratory \\
The University of Hong Kong\\
smyiu@cs.hku.hk
}

\begin{document}

\maketitle

\begin{abstract}
We set up a supply-side game-theoretic model for cooperative production of virtual products. In our model, a group of producers collaboratively produce a virtual product by contributing costly input resources to a production coalition. Producers are capacitated, i.e., they cannot contribute more resources than their capacity limits. Our model is an abstraction of emerging internet-based business models such as federated learning and crowd computing. To maintain an efficient and stable production coalition, the coordinator should share with producers the income brought by the virtual product. Besides the demand-side information asymmetry, another two sources of supply-side information asymmetry intertwined in this problem: 1) the capacity limit of each producer and 2) the cost incurred to each producer. In this paper, we rigorously prove that a supply-side mechanism from the VCG family, PVCG, can overcome such multiple information asymmetry and guarantee truthfulness. Furthermore,
%the famous Myerson-Satterthwaite impossibility result only holds true when there is a high degree of information asymmetry on both the demand side and the supply side. If 
with some reasonable assumptions, PVCG simultaneously attains truthfulness, ex-post allocative efficiency, ex-post individual rationality, and ex-post weak budget balancedness on the supply side, easing the well-known tension between these four objectives in the mechanism design literature.
\end{abstract}

\section{Introduction}
On this flat earth created by the internet, hierarchies of companies are falling away. Production activities are no longer a process confined within the border of an enterprise. Independent businesses can cooperate seamlessly through distributed computing systems to produce valuable virtual products. For example, a fast-developing technology, federated learning (FL), enables businesses to train (produce) artificial intelligence models collaboratively \cite{yang2019federated}. Another example is crowd computing, where participants contribute their redundant computing capacity to distributed computing tasks (e.g., protein folding simulation in the Folding@home project \cite{Beberg2009Folding}).
%To some extent, such cross-enterprise cooperative production coalitions are similar to a Cartel (a long-lasting form of monopoly organizations), yet they are socially beneficial, morally virtuous, and legally correct.
These production models differ from classical ones in at least two respects: First, collaborative production creates synergies, i.e., participants working together generate more value than they working separately; Second, the output is virtual products, which is non-tangible and non-rivalrous, i.e., the consumption of the virtual products by one participant does not reduce the amount available for others. While the first feature encourages cooperation, the second feature results in the free rider problem, i.e., participants lack the incentives to ally to attain the socially optimal result. Currently, cooperative production applications are mostly run as small-scale not-for-profit projects, heavily relying on volunteers' participation, thus hindering their popularization.
% Without a profitable business model, technology breakthroughs of cooperative production are hard to be widely adotped in the business world.
% The basic principle of economics assumes profit-seeking players.
A profitable business model is necessary for cooperative production to develop faster in the economic world. 

The design of such a business model is challenging due to two reasons: First, there are multiple dimensions of asymmetric information, i.e., the demand side is privately informed about its valuation of cooperative production while the supply side is privately informed about its cost and capacity constraint of supplying input resources; Second, in many realistic applications, an agent can both demand and supply cooperative production, e.g., businesses may both demand and supply data under federated learning. In this paper, we show that if the demand side's valuation satisfies a correlation condition, then we can use the combination of Cr{\'e}mer-McLean mechanism \cite{creemer1985optimal,kosenok2008individually}  and VCG mechanism to separately solve the asymmetric information problem for the demand side and supply side: on the demand side, the Cr{\'e}mer-McLean mechanism enables us to extract full consumer surplus for their usage of the virtual products and maximize the revenue distributable to the supply side; on the supply side, the VCG mechanism enables us to overcome the free-rider problem by compensating producers for their contribution of input resources and achieve socially efficient allocation. 
% According to business rules, value of output products should be shared with participants of the cooperative production process. This is crucial to the stability and sustainability of production coalitions. 
% A sharing rule, which determines how the value of the output product is shared among participants, is crucial to the stability and sustainability of production coalitions.

%In this paper, we set up a game-theoretic model for cooperative production of virtual products,
%In our game, independent capacipated producers participate in a cooperative production process by contributing costly input resources (e.g., labor, raw materials, equipment, data, etc.). The output virtual products bring value to consumers, bringing to the coalition monetary, which is shared among producers.
%The capacity limits, the costs of the input resources, and consumers' valuation on the output product are all unknown to the coordinator a priori. 
%based on which we propose a supply-side mechanism which determines payments to producers.
In this paper,  our contributions include: 1) We set up a game-theoretical model for cooperative production of virtual products, which takes into consideration twofold supply-side information asymmetry in addition to the demand-side information asymmetry; 2) We propose a procurement auction that incentivizes participants to report their private information honestly (truthfulness) and maximizes social welfare ex post (allocative efficiency); 3) We apply neural network methods to optimizing for ex-post individual rationality and ex-post weak budget balancedness; 4) We prove with some reasonable assumptions, truthfulness, ex-post allocative efficiency, ex-post individual rationality and ex-post weak budget balancedness can be attained simultaneously.

% It is a well-known fact that the above-mentioned desiderata are at least competing with each other, if not mutually exclusive. But with some reasonable assumptions on the cooperative production game, we attain all these goals simultaneously by introducing the Reverse-VCG (PVCG) mechanism, which is based on the famous VCG (Vickrey -Clarke-Groves) mechanism as proposed by \cite{vickrey1961counterspeculation,clarke1971multipart,groves1973incentives}). 

\section{Related Works}
Our game-theoretic model is an abstraction of emerging internet-based cooperative production practices. Although underlying technologies of cooperative production are fast developing, existing studies on these new business models are quite limited. From a more general perspective, our model is a special type of cooperative games \cite{peleg2007introduction,branzei2008models,chalkiadakis2011computational},
% In a cooperative game, players make binding agreements about the distribution of payoffs to ameliorate the Prisoner’s Dilemma. 
which have attracted great attention from the algorithmic game theory community \cite{sandholm1997coalitions,conitzer2006complexity,chalkiadakis2010cooperative,elkind2009computational,bachrach2009cost}.  
% In this subfield, a bunch of work has been devoted to coalition formation \cite{sandholm1997coalitions,conitzer2006complexity,chalkiadakis2010cooperative} and weighted voting game \cite{elkind2009computational,bachrach2009cost}. 
Designing an optimal sharing rule is a key concern when studying cooperative games \cite{kronbak2007sharing,niyato2011resource,weikard2009cartel}. 
% We also focus on the payoff sharing rule in this paper. We have gained enlightenment from previous work on Cartel stability \cite{thoron1998formation,weikard2009cartel}.

We decouples the cooperative production game into a demand-side game and a supply-side game. On the demand side, we use the Cr{\'e}mer-McLean mechanism \cite{creemer1985optimal,kosenok2008individually,albert2015assessing} to extract consumer surplus. On the supply side, we proposed the PVCG mechanism, which borrows the idea of procurement auction from procurement games \cite{chen2005efficient,chandrashekar2007auction,iyengar2008optimal,drechsel2010computing}. Compared to existing research, our model takes into account three sources of information asymmetry, one on the demand side and two on the supply side, yet we simultaneously achieve truthfulness, allocative efficiency, individual rationality, and ex-post weak budget balancedness, easing the well-known tension between these objectives \cite{jackson2014mechanism}.
% It is well-known that the complexity of mechanism design increases exponentially with more goals and more sources of information asymmetry. We reduce the complexity by following 
We also follow the recent trend of neural-network-based mechanism design \cite{shen2019automated}.

Federated learning \cite{yang2019federated,FL:2019} and crowd computing \cite{Beberg2009Folding,larson2009Folding} are examples of the cooperative production game under study. Therefore, our work is also related to the literature on training machine learning models with strategic participants \cite{jia2019towards,wang2019,cai2015optimum,richardson2019rewarding,westenbroek2019competitive,yu2020} and crowd computing from the angle of game theory \cite{christoforou2013crowd}.
% Different from previous case-by-case solutions, we tackle this problem by studying the more general cooperative production game, of which collaborative machine learning is a particular case. 

\section{Game Settings}
We study a cooperative production game where a set of $n$ \emph{producers}, denoted by $N=\{0,1,\ldots,n-1\}$, cooperatively produce a valuable \emph{virtual product}, of which copies are dilivered to $m$ \emph{consumers}, denoted by $M=\{n, n+1, \ldots, n+m-1 \}$. A participant may be both a producer and a consumer, but we assume that its behaviors as a producer and as a consumer are independent. On the \emph{supply side}, producers contribute input resources to the production coalition, e.g., labor, raw materials, equipments, etc. On the \emph{demand side}, consumers are granted access to utilize the output virtual product.
We introduce a parameter $x_i \ge \bm{0}$ to measure the input resources contributed by producer $i$. $x_i$ may be a vector when multiple input resources are involved. We use another parameter $y$ to measure the \emph{usefulness} of the output product, which is determined by the contributed input resources from all producers, i.e., $y$ is a function of $\bm{x} = (x_0, \ldots , x_{n-1})$. 

The usefulness of the output product determines the value it brings to consumers. We use the parameter $v_j, j\in N$ to denote the valuation of participant $j \in M$ on the output product. $v_j$ is a function of $y$ and a type parameter $\theta_j\in\Theta_j$ (called \emph{valuation type}) that reflects the heterogeneity among consumers. The type space is $\Theta_j=\{\theta_j^1,\cdots,\theta_j^{m_j}\}$ of cardinality $m_j<\infty$ and we assume that $0\in \Theta_j$ for all $j$. We denote this function by $w(\cdot)$ and call the composite function of $w(\cdot)$ and $y(\cdot)$ the \emph{individual valuation function}, denoted by $v(\cdot)$, i.e., $v(\bm{x},\theta_i)=w(y(\bm{x}),\theta_i)=v_i$.  For convenience, $\theta_j$ is such chosen that $v(\bm{x},\theta_j) \equiv 0$ when $\theta_j = 0$. We denote $\bm{\theta}=(\theta_n, \theta_{n+1},...,\theta_{n+m-1})$.
Contributing resources to the cooperative production process incurs costs to participants. Producer $i$'s cost $c_i = c(x_i, \gamma_i)$ is a function (called \emph{individual cost function}) of $x_i$ and another type parameter $\gamma_i$ that reflects the heterogeneity of producer $i$. We denote $\bm{\gamma}=(\gamma_0,...,\gamma_{n-1})$. In our game, producers are assumed to be capacitated, i.e., producer $i$ cannot contribute more resources than its capacity limit $\bar{x}_i$, i.e., $x_i \le \bar{x}_i$. Both the type parameters $\gamma_i, \theta_j$ and the capacity limit $\bar{x}_i$ are private information unknown to the coalition coordinator a priori. This coordinator makes the transfer payment $p_i$ to producer $i$ and $p_j$ to consumer $j$. Participants' preference is represented by quasi-linear utilities $u_i = p_i - c_i, \, i\in N$ and $u_j = p_j + v_j, \, j\in M$.
% Rational participants choose the strategy that maximizes their utilities. Meanwhile, 
We use \emph{social surplus} $S(\bm{x},\bm{\gamma})=\sum_{j=n}^{n+m-1} v(\bm{x},\theta_j)-\sum_{i=0}^{n-1} c(x_i,\gamma_i)$ to measure the social effect of the production coalition. Social surplus maximization implies \emph{Pareto efficiency} (or \emph{allocative efficiency} in the language of mechanism design).

%It is worth noting that this general game setting also applies to scenarios when producers and consumers are different entities. If participant $i$ is not a consumer, we set its valuation type $\theta_i = 0$. If participant $j$ is not a producer, we set its capacity limit $\bar{x}_j=0$.

Parameters in our model have concrete meanings in scenarios such as federated learning (FL) and crowd computing. In FL, the input resources are data. $x_i$ measures the size and quality of the contributed dataset. (Measuring the quality of datasets is a separate problem. We do not go deeper here.) The capacity limit $\bar{x}_i$ is the best dataset owned by participant $i$. $c_i$ is the cost of collecting and cleaning data. $v_i$ is the value of the federated model to each participant (e.g., for the use case that banks use FL to train AI models to predict credit risk, $v_i$ is calculated as the reduced bad debt rate times the principle of loans). In crowd computing, $x_i$ is the contributed computing power, $c_i$ costs of computing power (e.g., electricity costs, hardware costs, etc), and $v_j$ is the value of the computation result (e.g., the value of protein folding result to drug development).

We put forward five assumptions on the individual valuation function $v(\bm{x},\theta_j)$ and the individual cost function $c(x_i, \gamma_i)$. Assumption \ref{assumption:smooth}-\ref{assumption:zeroContribution} are used for proving truthfulness and alloactive efficiency. Assumption \ref{assumption:superAdditivity}-\ref{assumption:cremer} are used for proving individual rationality and ex-post weak budget balancedness.

\newtheorem{assumptionStart}{Assumption}
\begin{assumptionStart}[Smoothness and monotocity] \label{assumption:smooth}
    The individual valuation function $v(\bm{x}, \theta_j)$ is a smooth and monotonic increasing function of $\bm{x}$ and $\theta_j$.
    The individual cost function $c(x_i, \gamma_i)$ is a smooth and monotonic increasing function of $x_i$ and $\gamma_i$.
\end{assumptionStart}

\newtheorem{assumption2}[assumptionStart]{Assumption}
\begin{assumption2}[Zero input resource makes no difference] \label{assumption:zeroContribution}
    If $x_i = \bm{0}$, then:
    (1) participant $i$ makes no difference to the value of the output product, i.e.,
    $v(\bm{x}, \theta_j) = v((\bm{0},\bm{x}_{-i}),\theta_j) = v(\bm{x}_{-i},\theta_j), \forall j, \theta_j$;
    (2) participant $i$ bears no cost, i.e.,
    $c(x_i,\gamma_i) =c(\bm{0},\gamma_i) = 0, \, \forall \gamma_i$.
\end{assumption2}

\newtheorem{assumption3}[assumptionStart]{Assumption}
\begin{assumption3}[Super additivity] 
    \label{assumption:superAdditivity}
    The individual valuation function $v(\bm{x}, \theta_j)$ is super additive with respect to $x_i, i=0,\ldots,n-1$, i.e.,
    \begin{equation}
        \label{equation:superAdditivity}
       v(\bm{x}, \theta_j)
       \ge
       \sum_{i=0}^{n-1}v(x_i, \theta_j),
        \forall \bm{x}, \theta_j.
    \end{equation}
\end{assumption3}

\newtheorem{assumption4}[assumptionStart]{Assumption}
\begin{assumption4}[Decreasing cross marginal returns]
    \label{assumption:decreasingMarginal}
    The marginal return of one producer's input resources decreases when other producers contribute more input resources, i.e.,
    for all $i\in N,j\in M,\theta_j$ and $x_i \ge x_i'$, $\bm{x}_{-i} \ge \bm{x}_{-i}'$,
    the following inequality holds:
    \begin{align}
       \label{equation:decreasingMarginal}
       &v((x_i,\bm{x}_{-i}), \theta_j)
       -
       v((x_i',\bm{x}_{-i}), \theta_j)
       \nonumber \\
       \le
       &
       v((x_i,\bm{x}_{-i}'), \theta_j)
       -
       v((x_i',\bm{x}_{-i}'), \theta_j).
    \end{align}
\end{assumption4}

\newtheorem{assumption5}[assumptionStart]{Assumption}
\begin{assumption5}[Correlated and identifiable valuation types]
    \label{assumption:cremer}
    The prior belief on valuation types, Prior($\bm{\theta}$), is identifiable and correlated, i.e., the identifiability condition \cite{kosenok2008individually} and Cr{\'e}mer-McLean condition \cite{creemer1985optimal} hold for all consumers. 
\end{assumption5}

The economic meaning of super additivity is that cooperative production brings synergies, i.e., the value created by the production coalition is higher than the total value created by independent producers. The rationality of Assumption \ref{assumption:decreasingMarginal} comes from the law of diminishing marginal returns: when many resources have already been involved in a production process, the marginal return brought by an additional unit of input resources decreases. As an example, the following individual valuation function satisfies all these four assumptions:
\begin{equation}\label{eqaution:exampleValueFunction}
    v(\bm{x}, \theta_j) = \theta_j \sqrt{n\sum_{k=0}^{n-1}z(x_k)^2},
\end{equation}
where $z(x_i)$ is an arbitrary smooth and increasing function (e.g., the Cobb-Douglas function).
Assumption \ref{assumption:cremer} is introduced to guarantee full consumer surplus, so that we can derive the income of the coalition $-\sum_{j=n}^{n+m-1}p_j$ from individual valuation functions, i.e., $-\sum_{j=n}^{n+m-1}p_j=\sum_{j=n}^{n+m-1}v(\bm{x}, \theta_j)$, because we have the following theorem:
\newtheorem{theoremstart}{Theorem}
\begin{theoremstart} [Cr{\'e}mer-McLean Theorem]\label{theorem:CremerMcLean}
    There exists an interim individually rational and ex-post budget balanced Bayesian mechanism that extracts full consumer surplus if Prior($\bm{\theta}$) is identifiable and Cr{\'e}mer-McLean condition holds for all consumers.
\end{theoremstart}

Proofs of Theorem \ref{theorem:CremerMcLean} can be found in \cite{creemer1985optimal} and \cite{kosenok2008individually}. We can construct such a demand-side \emph{Cr{\'e}mer-McLean mechanism} by following the constructive proof of Lemma A3 in \cite{kosenok2008individually} or by automated mechanism design techniques \cite{albert2015assessing}. We will not go into further detail here because we focus on the supply-side game. For theoretical analyses in this paper, we take as a given that the production coalition uses Cr{\'e}mer-McLean mechanism to extract full consumer surplus. This also guarantees the reported valuation type $\hat{\theta}_j$ equals the true valuation type $\theta_j$ for all consumer $j\in M$. Theorem \ref{theorem:CremerMcLean} transfers the ex-post weak budget balance constraint $\sum_{l=0}^{n+m-1}p_l \le 0$ to the following inequality:
\begin{equation}\label{eqaution:WBBgivenCremer}
    \sum_{i=0}^{n-1}p_i \le \sum_{j=n}^{n+m-1}v(\bm{x}, \theta_j).
\end{equation}
This decouples the supply-side mechanism design problem from the demand-side problem.

% $p_j=-v(\bm{x}, \theta_j)+\alpha[\mathbb{E}_{\theta_j}\text{ln} \mu (\bm{\theta}_{-j}|\theta_j)- \text{ln} \mu (\bm{\theta}_{-j}|\theta_j) ]$

\section{The Procurement Auction}

As a counterpart of Cr{\'e}mer-McLean mechanism which is optimal on the demand side, we introduce an optimal supply-side procurement auction in this section. This proposed procurement auction, accompnied by the demand-side Cr{\'e}mer-McLean mechanism, maximizes social surplus by incentivizing producers to truthfully report their capacity limits and type parameters. This procurement auction consists of four steps, of which the soundness will be proved in the next section.
%The primary goal of the procurement auction is to induce truthful information on the capacity limit $\bar{x}_i$ and the type parameters $\gamma_i, \theta_i$ from every participant. 

\subsection{Step 1. Producers bid on capacity limits and cost types}

As the first step, every producer submits a sealed bid for their respective capacity limits and cost types. The reported capacity limit $\hat{x}_i$ is the maximum resources that producer $i$ is willing to offer to the coalition. It may differ from the true capacity limit $\bar{x}_i$. Similarly, the reported cost type $\hat{\gamma}_i$ may differ from the true cost type $\gamma_i$. 
% It is worth noting that we maximize social surplus based on reported type parameters $\hat{\gamma_i}, 0,\ldots,n-1$, but the realized social surplus is determined by the true type parameters $\gamma_i, i = 0,\ldots,n-1$. A truthful mechanism guarantees the equivalence between reported parameters and truthful parameters.
% Therefore, we need to guarantee truthful reporting of these parameters.

\subsection{Step 2. The coordinator chooses the optimal acceptance ratios}

Then, the coalition coordinator decides how many input resources to accept from each producer. It chooses $x_i \le \hat{x}_i, i=0,\ldots,n-1$ that maximize the social surplus, constrained by reported capacity limits and based on reported type parameters. Equivalently, the coordinator calculates the optimal \emph{acceptance ratio} $\eta_i \in [0,1]^{\textrm{dim}(x_i)} = x_i \oslash \hat{x}_i$ such that $x_i = \hat{x}_i \odot \eta_i$, where $\odot$ and $\oslash$ denote the element-wise multiplication and division respectively, and $[0,1]$ denotes the interval between $0$ and $1$. The economic meaning of $\eta_i$ is the ratio of input resources accepted by the coalition to those offered by the producer.

The optimal acceptance ratios $(\eta^*_0, \ldots, \eta^*_{n-1})=\bm{\eta}^*$ are calculated according to the following formula:
\begin{align}\label{equation:computeEta}
    &\bm{\eta}^{*} =
     \textrm{argmax}_{
        \bm{\eta} \in [0,1]^{\textrm{dim}(x_i) \times n }
        }
    \{
        S(\hat{\bm{x}}\odot\bm{\eta}, \hat{\bm{\gamma}}, \hat{\bm{\theta}})
    \}
    \\
    &= \mathop{\textrm{argmax}}\limits_{
        \bm{\eta} \in [0,1]^{\textrm{dim}(x_i) \times n }
        }
    \sum_{j=n}^{n+m-1} v(\hat{\bm{x}}\odot\bm{\eta},\theta_j)-\sum_{i=0}^{n-1}c_i(\hat{x}_i \odot \eta_i, \hat{\gamma}_i).
    \nonumber
\end{align}
Because different $(\hat{\bm{x}},\hat{\bm{\gamma}}, \hat{\bm{\theta}})$ results in different $\bm{\eta}^{*}$, $\bm{\eta}^{*}$ is written as $\bm{\eta}^{*}(\hat{\bm{x}},\hat{\bm{\gamma}}, \hat{\bm{\theta}})$. Correspondingly, the maximum social surplus is denoted by
$S^{*}(\hat{\bm{x}},\hat{\bm{\gamma}}, \hat{\bm{\theta}})= \sum_{j=n}^{n+m-1}
v(\hat{\bm{x}}\odot\bm{\eta}^{*}(\hat{\bm{x}},\hat{\bm{\gamma}},\hat{\bm{\theta}}), \hat{\theta}_j) -\sum_{i=0}^{n-1}c_i(\hat{x}_i
\odot \eta_i^{*}(\hat{\bm{x}},\hat{\bm{\gamma}},\hat{\bm{\theta}}), \hat{\gamma}_i)$. It is worth noting that although
$S^{*}(\hat{\bm{x}},\hat{\bm{\gamma}}, \hat{\bm{\theta}})$ and
$S(\bm{x},\bm{\gamma},\bm{\theta})$ both represent social surplus, they are different functions. The first parameter  $\bm{x}$ in $S(\cdot)$ is the accepted input resources, whereas the first parameter  $\hat{\bm{x}}$ in $S^{*}(\cdot)$ is the reported capacity limits. $\bm{x}$ and $\hat{\bm{x}}$ are related by $\bm{x} = \hat{\bm{x}} \odot \bm{\eta}^{*} $.

\subsection{Step 3. Producers contribute accepted input resources to the production coalition}

In this step, producers are required to contribute $\hat{\bm{x}} \odot \bm{\eta}^*$ units of input resources to the production coalition. Since in the first step, producer $i$ has promised to offer at most $\hat{x}_i$ units of input resources, if it cannot contribute $\hat{x}_i \odot \eta^*_i \le \hat{x}_i$, we impose a high punishment on it.
% , e.g., restricting it from joining subsequent production coalitions. When such punishment is high enough, We can prove that participants contribute to the coalition exactly the accepted quantity of input resources, collaboratively producing the output products and generating the payoff $V(\hat{\bm{x}} \odot \bm{\eta}^*)$.
With the contributed input resources, producers collaboratively produce the output virtual product, bringing value $v(\hat{\bm{x}} \odot \bm{\eta}^*,\theta_j)$ to consumer $j\in M$.

\subsection{Step 4. The coordinator makes transfer payments to participants according to the PVCG sharing rule}

In this final step, the coordinator pays producers according to the PVCG sharing rule. The PVCG payment
\begin{equation} \label{equation:totalPayment}
    p_i(\cdot) = \tau_i(\cdot) + h^*_i(\cdot)
\end{equation}
is composed of two parts, the \emph{VCG payment} $\tau_i$ and the \emph{optimal adjustment payment} $h^*_i$. The VCG payment is designed to induce truthfulness, i.e., the reported capacity limits $\hat{\bm{x}}$ and reported cost type $\hat{\bm{\gamma}}$ are equal to the true capacity limits $\bar{\bm{x}}$ and true cost type $\bm{\gamma}$. The adjustment payment is optimized so that ex-post individual rationality and ex-post weak budget balancedness can also be attained.
%Proofs will be provided in the next section.

With $\bm{\eta}^{*}$ calculated in Step 2, the VCG payment $\tau_i$ to producer $i$ is:
\begin{align}\label{equation:computeTau}
    &\tau_i =S^{*}(\hat{\bm{x}},\hat{\bm{\gamma}}, \hat{\bm{\theta}}_{-i})
    -S^{*}_{-i}(\hat{\bm{x}}_{-i},\hat{\bm{\gamma}}_{-i},\hat{\bm{\theta}})
    \nonumber \\
    & \quad
    +c(\hat{x}_i\odot\eta_{i}^{*}(\hat{\bm{x}},\hat{\bm{\gamma}},\hat{\bm{\theta}}),\hat{\gamma}_i) \nonumber \\
    & 
    = \sum_{j=n}^{n+m-1}[ v(\hat{\bm{x}}\odot\bm{\eta}^{*}(\hat{\bm{x}},\hat{\bm{\gamma}},\hat{\bm{\theta}}), \hat{\theta}_j)-
    v(\hat{\bm{x}}_{-i}\odot\bm{\eta}^{-i*}(
    \nonumber \\
    & \quad    
    \hat{\bm{x}}_{-i},\hat{\bm{\gamma}}_{-i},\hat{\bm{\theta}}), \hat{\theta}_j)]
   - \sum_{k=0, \neq i}^{n-1}[c(\hat{x}_k \odot \eta_{k}^{*}(\hat{\bm{x}},\hat{\bm{\gamma}},\hat{\bm{\theta}}),\hat{\gamma}_k)
     \nonumber \\
    & \quad    
    - c(\hat{x}_k \odot \eta_{k}^{-i*}(\hat{\bm{x}}_{-i},\hat{\bm{\gamma}}_{-i},\hat{\bm{\theta}}),\hat{\gamma}_k)], 
\end{align}
where $(\hat{\bm{x}}_{-i},\hat{\bm{\gamma}}_{-i})$ denotes the reported capacity limits and the reported cost types excluding producer $i$. $\bm{\eta}^{-i*}$ and $S^{*}_{-i}(\hat{\bm{x}}_{-i},\hat{\bm{\gamma}}_{-i},\hat{\bm{\theta}})$ are the corresponding optimal acceptance ratios and maximum social surplus. Note that $\bm{\eta}^{-i*}$ is different from $\bm{\eta}^{*}_{-i}$: the former maximizes $S(\hat{\bm{x}}_{-i} \odot\bm{\eta}_{-i},\hat{\bm{\gamma}}_{-i},\hat{\bm{\theta}})$, whereas the latter is the component of $\bm{\eta}^{*}$ that maximizes
$S(\hat{\bm{x}} \odot \bm{\eta},\hat{\bm{\gamma}},\hat{\bm{\theta}})$. $\bm{\tau} = (\tau_0, \ldots, \tau_{n-1})$ is a function of $(\hat{\bm{x}},\hat{\bm{\gamma}},\hat{\bm{\theta}})$, written as $\bm{\tau}(\hat{\bm{x}},\hat{\bm{\gamma}},\hat{\bm{\theta}})$. 

The \emph{adjustment payment} $h_i(\hat{\bm{x}}_{-i}, \hat{\bm{\gamma}}_{-i},\hat{\bm{\theta}})$ is a function of $(\hat{\bm{x}}_{-i}, \hat{\bm{\gamma}}_{-i},\hat{\bm{\theta}})$. 
% In PVCG, we use the \emph{optimal adjustment payment}  $h_i^*(\cdot)$. 
The optimal adjustment payments $(h^*_0 (\cdot),\ldots,h^*_{n-1} (\cdot))=\bm{h}^*(\cdot)$ are determined by solving the following \emph{functional equation} (a type of equation in which the unknowns are functions instead of variables; refer to \cite{rassias2012functional} for more details):
\begin{align} \label{equation:calculateAdjustment}
    & \sum_{i=0}^{n-1} \textrm{ReLu}[
        - (S^{*}(\bm{x},\bm{\gamma},\bm{\theta})
        -S^{*}_{-i}(\bm{x}_{-i},\bm{\gamma}_{-i},\bm{\theta}))
    \nonumber \\
    &
    -h_i(\bm{x}_{-i},\bm{\gamma}_{-i},\bm{\theta}) 
    ]
    + \textrm{ReLu}
    \{ \sum_{i=0}^{n-1}[(S^{*}(\bm{x},\bm{\gamma},\bm{\theta})
    \nonumber
    \\
    &
    -S^{*}_{-i}(
        \bm{x}_{-i},\bm{\gamma}_{-i},\bm{\theta}))
    + h_i(\bm{x}_{-i},\bm{\gamma}_{-i},\bm{\theta}) ]
    \\
    &
    -S^{*}(\bm{x},\bm{\gamma},\bm{\theta})\}
    \equiv 0, \, \forall
    (\bar{\bm{x}},\bm{\gamma},\bm{\theta}) \in \textrm{supp}(\textrm{Prior}(\bm{x},\bm{\gamma}, \bm{\theta})),
    \nonumber
\end{align}
where $\textrm{supp}(\textrm{Prior}(\bm{x},\bm{\gamma}, \bm{\theta}))$ is the \emph{support} of the \emph{prior distribution} $\textrm{Prior}(\bm{x},\bm{\gamma}, \bm{\theta})$ of the true parameters $(\bm{x},\bm{\gamma}, \bm{\theta})$ estimated by the coalition coordinator. Support is a terminology from \emph{measure theory}, defined by
    $\textrm{supp}(\textrm{Prior}(\bm{x},\bm{\gamma}, \bm{\theta}))
    = \{(\bm{x},\bm{\gamma}, \bm{\theta})|\textrm{Prior}(\bm{x},\bm{\gamma}, \bm{\theta})>0\}$. 
% In Eq. \ref{equation:calculateAdjustment}, we use $\bm{x}$ and $\bm{\gamma}$ instead of $\hat{\bm{x}}$ and $\hat{\bm{\gamma}}$ because of truthfulness guaranteed by the VCG payments.
In general, there is no closed-form solution to Eq. \ref{equation:calculateAdjustment}, so we employ neural network techniques to learn the solution. We will go into more details later.

\section{Theoretical Analyses} \label{sec:theory}
Theoretical analyses in this section are organized through the following strand:
\begin{enumerate}
    \item First (in Proposition \ref{proposition:IC} and Proposition \ref{proposition:SSM}), we prove that for arbitrary $h_i(\hat{\bm{x}}_{-i},\hat{\bm{\gamma}}_{-i},\hat{\bm{\theta}})$, $i=0,\ldots,n-1$, the PVCG payments guarantee supply-side truthfulness (dominant strategy incentive compatibility) and maximize social surplus (allocative efficiency).
    \item Second (in Proposition \ref{proposition:IR} and Proposition \ref{proposition:WBB}), given that Cr{\'e}mer-McLean Theorem (Theorem \ref{theorem:CremerMcLean}) holds on the demand side, we derive two ineqaulity cosntraints on the adjustment payment $h_i(\hat{\bm{x}}_{-i},\hat{\bm{\gamma}}_{-i},\hat{\bm{\theta}})$, which are sufficient and necessary conditions for ex-post individual rationality and ex-post weak budget balancedness. We show that these constraints can be transformed into Eq. \ref{equation:calculateAdjustment}, which is equivalent to a minimization problem that can be solved by neural network methods.
    \item Lastly (in Theorem \ref{theorem:WBBIRtogether} and Corollary \ref{corollary:superAdditivityLeadsToCorollary2}), we prove the existence of at least one solution $\bm{h}^*(\cdot)$ to Eq. \ref{equation:calculateAdjustment}. The PVCG payment corresponding to this solution attains truthfulness, ex-post allocative efficiency, ex-post individual rationality, and ex-post weak budget balancedness simultaneously on the supply side.
\end{enumerate}

First, we prove that for arbitrary $h_i(\hat{\bm{x}}_{-i},\hat{\bm{\gamma}}_{-i},\hat{\bm{\theta}})$, the payment
$p_i(\cdot) = \tau_i(\cdot) + h_i(\cdot)$ encourages all producers to report their capacity limits and cost types truthfully.

\newtheorem{propositionStart}{Proposition}
\begin{propositionStart}[Dominant strategy incentive compatibility] \label{proposition:IC}
    For every producer $i$, truthfully reporting its capacity limit $\bar{x}_i$ and cost type $\gamma_i$ is its dominant strategy, i.e.,
    \begin{align}\label{equation:DIC}
       & p_i(
            (\bar{x}_i,\hat{\bm{x}}_{-i}),
            (\gamma_i, \hat{\bm{\gamma}}_{-i}),
            \hat{\bm{\theta}}
            )
        \nonumber \\
        & - c(
                \bar{x}_i \odot \eta^{*}_i(
                (\bar{x}_i,\hat{\bm{x}}_{-i}),
            (\gamma_i, \hat{\bm{\gamma}}_{-i}),
            (\theta_i, \hat{\bm{\theta}}_{-i})
            ),
                \gamma_i
        )
                \nonumber \\
        & 
        \ge
        p_i
            (\hat{\bm{x}},\hat{\bm{\gamma}}, \hat{\bm{\theta}})
        - c_i(
            \hat{x}_i \odot \eta^{*}_i
            (\hat{\bm{x}},\hat{\bm{\gamma}},\hat{\bm{\theta}}),
            \gamma_i
        )
         \nonumber \\
        & 
        \quad , \forall i\in N, x_i,\hat{\bm{x}},
        \gamma_i,\hat{\bm{\gamma}},
        \hat{\bm{\theta}}.
    \end{align}
\end{propositionStart}

\begin{proof}
    We differentiate the cases where $\hat{x}_i \odot \eta^*_i(\hat{\bm{x}},\hat{\bm{\gamma}},\hat{\bm{\theta}}) > \bar{x}_i$ with those where $\hat{x}_i \odot \eta^*_i (\hat{\bm{x}},\hat{\bm{\gamma}},\hat{\bm{\theta}})\le \bar{x}_i$.

    (1) When $\hat{x}_i \odot \eta^*_i(\hat{\bm{x}},\hat{\bm{\gamma}},\hat{\bm{\theta}}) > \bar{x}_i$,
    \\
    it is impossible for producer $i$ to contribute the accepted amount of input resources $\hat{x}_i \odot \eta^*_i(\hat{\bm{x}},\hat{\bm{\gamma}},\hat{\bm{\theta}})$ because this exceeds its capacity limit. In this case, producer $i$ suffers a high punishment, i.e., $p_i
            (\hat{\bm{x}},\hat{\bm{\gamma}}, \hat{\bm{\theta}})$ is a large negative number.
    % However, it cannot receive higher payment because according to the PVCG payment in Eq. \ref{equation:totalPayment}, (1) the adjustment payment $h_i$ is independent of $\hat{x}_i$ (2) and the VCG payment only depends on the accepted input resources $\hat{x}_i \odot \eta^*_i$, on which it cannot fulfil its promise in this case.
    % Since the punishment is assumed to be high enough, 
    Hence, the right side of Eq. \ref{equation:DIC} becomes extremely negative and Eq. \ref{equation:DIC} holds. This shows a rational producer will not choose $\hat{x}_i$ such that $\hat{x}_i \odot \eta^*_i(\hat{\bm{x}},\hat{\bm{\gamma}},\hat{\bm{\theta}}) > \bar{x}_i$.

    (2) When $\hat{x}_i \odot \eta^*_i (\hat{\bm{x}},\hat{\bm{\gamma}},\hat{\bm{\theta}}) \le \bar{x}_i$,
    \\
    we aim to prove that by truthfully reporting $\bar{x}_i$, $\gamma_i$, and $\theta_i$, the utility of producer $i$ is at least the same as before.

    By definition, \\
    $S^{*}(\hat{\bm{x}},\hat{\bm{\gamma}},\hat{\bm{\theta}}) =
    \textrm{max}_{
        \bm{\eta} \in [0,1]^{\textrm{dim}(x_i) \times n }
        }
    \{
        \sum_{j=n}^{n+m-1}v(\hat{\bm{x}}\odot\bm{\eta},\hat{\theta}_j)- \sum_{i=0}^{n-1}c(\hat{x}_i \odot \eta_i, \hat{\gamma}_i)
    \}, \forall \hat{\bm{x}},\hat{\bm{\gamma}},\hat{\bm{\theta}}$. We substitute $(\hat{\bm{x}},\hat{\bm{\gamma}},\hat{\bm{\theta}})$ with $((\bar{x}_i,\hat{\bm{x}}_{-i}),(\gamma_i, \hat{\bm{\gamma}}_{-i}),\hat{\bm{\theta}})$ and get
    \begin{align}\label{equation:DICproof1}
    &
        \quad
        S^{*}(
            (\bar{x}_i,\hat{\bm{x}}_{-i}),
            (\gamma_i, \hat{\bm{\gamma}}_{-i}),\hat{\bm{\theta}}
        )
        \nonumber \\
        &
        \ge
        \sum_{j=n}^{n+m-1}
        v((\bar{x}_i,\hat{\bm{x}}_{-i})\odot\bm{\eta},\theta_j)
        -c(\bar{x}_i \odot \eta_i, \gamma_i)
        \nonumber \\
        & \quad
        -  \sum_{k=0, \ne i}^{n}
        c(\hat{x}_k \odot \eta_k, \hat{\gamma}_k)
        , \, \forall \bm{\eta}
        \in
        [0,1]^{\textrm{dim}(x_i) \times n}.
    \end{align}
    In particular, for
    $\bm{\eta} =
    (\hat{x}_i \odot \eta^{*}_i
    (\hat{\bm{x}},\hat{\bm{\gamma}},\hat{\bm{\theta}})
    \oslash \bar{x}_i,
    \bm{\eta}^{*}_{-i}
    (\hat{\bm{x}},\hat{\bm{\gamma}},\bm{\theta})
    )
    \in
    [0,1]^{\textrm{dim}(x_i) \times n}$, Eq. \ref{equation:DICproof1} holds. Therefore,
    \begin{align}\label{equation:DICproof2}
    &  S^{*}  (
            (\bar{x}_i,\hat{\bm{x}}_{-i}),
            (\gamma_i, \hat{\bm{\gamma}}_{-i}),\hat{\bm{\theta}}
        )
        \nonumber \\
    \ge &
        \sum_{j =n}^{n+m-1} [v(
            (\bar{x}_i,\hat{\bm{x}}_{-i}) \odot (\hat{x}_i \odot \eta^{*}_i
            (\hat{\bm{x}},\hat{\bm{\gamma}},\hat{\bm{\theta}})
            \oslash \bar{x}_i
              \nonumber \\
    & ,
            \bm{\eta}^{*}_{-i}
            (\hat{\bm{x}},\hat{\bm{\gamma}},\hat{\bm{\theta}})
            )
        ,\hat{\theta}_j)
        - \sum_{k=0, \ne i}^{n-1}  c(
            \hat{x}_k \odot \eta^{*}_k(\hat{\bm{x}},
            \hat{\bm{\gamma}},\hat{\bm{\theta}}),
            \hat{\gamma}_k
         )
        \nonumber \\
    &
        - c(
            \bar{x}_i \odot
            \hat{x}_i \odot \eta^{*}_i
            (\hat{\bm{x}},\hat{\bm{\gamma}},\hat{\bm{\theta}})
            \oslash \bar{x}_i,
            \gamma_i
        )
    \nonumber     \\
    % = &
    %     B(
    %         (\check{d}_i \odot
    %         \eta^{*}_i
    %         (\check{\bm{d}},\hat{\bm{\gamma}})
    %         ,
    %         \check{\bm{d}}_{-i} \odot
    %         \bm{\eta}^{*}_{-i}
    %         (\check{\bm{d}},\hat{\bm{\gamma}})
    %         )
    %     )
    %     \nonumber \\
    % &   -
    %     \sum_{k \ne i} c_i(
    %         \check{d}_k \odot \eta^{*}_k(\check{\bm{d}},
    %         \hat{\bm{\gamma}}),
    %         \hat{\gamma}_k
    %     )
    %     \nonumber \\
    % &
    %     - c_i(
    %         \check{d}_i \odot \eta^{*}_i
    %         (\check{\bm{d}},\hat{\bm{\gamma}}),
    %         \gamma_i
    %     ) +
    %     g(\check{d}_i)
    %     \nonumber \\
    %%%%%%%%%%%%%%%%%%%%%%%% page limit%%%%%%%%%%%%
    % = &
    %     S (
    %     (\hat{x}_i,\hat{\bm{x}}_{-i}) \odot
    %      \bm{\eta}^{*}
    %         (\hat{\bm{x}},\hat{\bm{\gamma}},\hat{\bm{\theta}}),
    %     (\hat{\gamma}_i, \hat{\bm{\gamma}}_{-i}),
    %       (\hat{\theta}_i, \hat{\bm{\theta}}_{-i})
    % )
    %         \nonumber \\
    % &
    %     - v(
    %         (\hat{x}_i,\hat{\bm{x}}_{-i}) \odot
    %      \bm{\eta}^{*}
    %         (\hat{\bm{x}},\hat{\bm{\gamma}},\hat{\bm{\theta}}),
    %         \hat{\theta}_i
    %     )
    %     \nonumber \\
    % &
    %     + v(
    %         (\hat{x}_i,\hat{\bm{x}}_{-i}) 
    %         \odot
    %      \bm{\eta}^{*}
    %         (\hat{\bm{x}},\hat{\bm{\gamma}},\hat{\bm{\theta}}),
    %         \theta_i
    %     )
    % \\
    % &
    %     + c(
    %         \hat{x}_i \odot \eta^{*}_i
    %         (\hat{\bm{x}},\hat{\bm{\gamma}},\hat{\bm{\theta}}),
    %         \hat{\gamma}_i
    %     )
    %     - c(
    %         \hat{x}_i \odot \eta^{*}_i
    %         (\hat{\bm{x}},\hat{\bm{\gamma}},\hat{\bm{\theta}}),
    %         \gamma_i
    %     )
    % \nonumber \\
    %%%%%%%%%%%%%%%%%%%%%%%%%%%%%%%%%%%%%%%%%%%%%%
    = & S^{*}  (\hat{\bm{x}},\hat{\bm{\gamma}},\hat{\bm{\theta}})
    \\
    &
    + c(
        \hat{x}_i \odot \eta^{*}_i
        (\hat{\bm{x}},\hat{\bm{\gamma}},\hat{\bm{\theta}}),
        \hat{\gamma}_i
    )
    - c(
        \hat{x}_i \odot \eta^{*}_i
        (\hat{\bm{x}},\hat{\bm{\gamma}},\hat{\bm{\theta}}),
        \gamma_i
    ).
    \nonumber
    \end{align}
    Adding $h_i(\hat{\bm{x}}_{-i},\hat{\bm{\gamma}}_{-i},\hat{\bm{\theta}}_{-i})
    - S^{*}_{-i}(\hat{\bm{x}}_{-i},
    \hat{\bm{\gamma}}_{-i},\hat{\bm{\theta}}_{-i})$
    to both sides of Eq. \ref{equation:DICproof2} and substituting
    $p_i(\cdot) + v(\cdot) - c(\cdot)=S^{*}(\cdot) - S_{-i}^{*}(\cdot) + h(\cdot)$, we get Eq. \ref{equation:DIC}.
\end{proof}

From Proposition \ref{proposition:IC}, we know reported parameters equal true parameters, i.e., $(\hat{\bm{x}}, \hat{\bm{\gamma}},\hat{\bm{\theta}}) = (\bar{\bm{x}}, \bm{\gamma},\bm{\theta})$. Therefore, we can use these two sets of parameters interchangeably in the remaining parts of this paper.

\newtheorem{proposition2}[propositionStart]{Proposition}
\begin{proposition2}[Ex-post social surplus maximization / allocative efficiency] \label{proposition:SSM}
    PVCG maximizes social surplus ex post.
\end{proposition2}

\begin{proof}
    Suppose $\bm{x}^{**} = \textrm{argmax}_{\bm{x} \le \bar{\bm{x}}} \{ S(\bm{x},\bm{\gamma}, \bm{\theta})\}$
    and $S^{**} = S(\bm{x}^{**},\bm{\gamma},\bm{\theta})$.
    We aim to prove that PVCG results in ex-post social surplus no less than $S^{**}$.

    By definition of $\bm{\eta}^{*}(\hat{\bm{x}}, \hat{\bm{\gamma}},\hat{\bm{\theta}})$,
    \begin{equation} \label{equation: propositionSSMproof1}
    S(\hat{\bm{x}} \odot \bm{\eta}^{*}(\hat{\bm{x}}, \hat{\bm{\gamma}},\hat{\bm{\theta}}), \hat{\bm{\gamma}}, \hat{\bm{\theta}}) \ge S(\hat{\bm{x}} \odot \bm{\eta}, \hat{\bm{\gamma}}, \hat{\bm{\theta}}),\, \forall \bm{\eta}.
    \end{equation}
    Incentive compatibility guarantees $\hat{\bm{x}} = \bar{\bm{x}}$, $\hat{\bm{\gamma}}=\bm{\gamma}$, and $\hat{\bm{\theta}}=\bm{\theta}$.
    Therefore,
    % Eq. \ref{equation: propositionSSMproof1} becomes
    \begin{equation} \label{equation: propositionSSMproof2}
    S(\hat{\bm{x}} \odot \bm{\eta}^{*}(\hat{\bm{x}}, \hat{\bm{\gamma}},\hat{\bm{\theta}}), \bm{\gamma},\bm{\theta}) \ge S(\bar{\bm{x}} \odot \bm{\eta}, \bm{\gamma},\bm{\theta}),\, \forall \bm{\eta}.
    \end{equation}
    Particularly, Eq. \ref{equation: propositionSSMproof2} holds for
    $\bm{\eta} = \bm{x}^{**} \oslash \bar{\bm{x}} \in [0,1]^{\textrm{dim}(x_i) \times n}$, i.e.,
    \begin{align} \label{equation: propositionSSMproof3}
    S(\hat{\bm{x}} \odot \bm{\eta}^{*}(\hat{\bm{x}}, \hat{\bm{\gamma}},\hat{\bm{\theta}}), \bm{\gamma},\bm{\theta})
    & \ge
    S(\bar{\bm{x}} \odot \bm{x}^{**} \oslash \bar{\bm{x}}, \bm{\gamma},\bm{\theta})
    \nonumber \\
    & =
    S(\bm{x}^{**}, \bm{\gamma},\bm{\theta})
    = S^{**}.
    \end{align}
    The left side of Eq. \ref{equation: propositionSSMproof3} is the ex-post social surplus achieved by PVCG, while the right side is the maximum social surplus across all possible $\bm{x} \le \bar{\bm{x}}$, given $(\bar{\bm{x}},\bm{\gamma},\bm{\theta})$. Hence, the right side is also no less than the left side. Therefore, the left side equals the right side.
\end{proof}

Furthermore, given that incentive compatibility has been proved in Proposition \ref{proposition:IC} and Theorem \ref{theorem:CremerMcLean}, the following two propositions provide sufficient and necessary conditions for ex-post individual rationality and ex-post weak budget balancedness.

\newtheorem{proposition3}[propositionStart]{Proposition}
\begin{proposition3}[Condition for ex-post individual rationality] \label{proposition:IR}
    PVCG is ex-post individual rational (IR) for all producers i.f.f.
    the true capacity limits $\bar{\bm{x}}$, the true type parameters $\bm{\gamma},\bm{\theta}$, and the adjustment payments $\bm{h}(\cdot)$ satisfy
    \begin{align}\label{equation:IRpremise}
         h_i & (
             \bar{\bm{x}}_{-i},
             \bm{\gamma}_{-i},
             \bm{\theta}
         ) 
        \\
        & 
         \ge
         - [
             S^{*}(\bar{\bm{x}},
             \bm{\gamma},\bm{\theta}) - S^{*}_{-i} (
                 \bar{\bm{x}}_{-i},
                 \bm{\gamma}_{-i},
                 \bm{\theta}
             )
         ], \, \forall i \in N.
         \nonumber
    \end{align}
\end{proposition3}

\begin{proof}
    According to truthfulness proved in Proposition \ref{proposition:IC} and Theorem \ref{theorem:CremerMcLean},
    we use $\bar{\bm{x}}, \bm{\gamma}, \bm{\theta}$ to substitute $\hat{\bm{x}},\hat{\bm{\gamma}}, \hat{\bm{\theta}}$ in Eq. \ref{equation:totalPayment} and
    Eq. \ref{equation:computeTau}.
    Then, the ex-post utility of producer $i$ becomes
    \begin{align}\label{equation:IRproof}
    & u_i(\bar{\bm{x}},\bm{\gamma},\bm{\theta}) 
    =p_i(\bar{\bm{x}},\bm{\gamma},\bm{\theta})
    - c(\bar{x}_i \odot \eta_{i}^{*}(\bar{\bm{x}},\bm{\gamma},\bm{\theta}),\gamma_i)
    \\
    &=S^{*}(\bar{\bm{x}},\bm{\gamma},\bm{\theta}) -S^{*}_{-i}(\bar{\bm{x}}_{-i},\bm{\gamma}_{-i},\bm{\theta})
        +h_i(\bar{\bm{x}}_{-i},\bm{\gamma}_{-i},\bm{\theta}) .
    \nonumber
    \end{align}
    Ex-post IR requires $u_i(\bar{\bm{x}},\bm{\gamma},\bm{\theta}) \ge 0, \forall i$, which is equivalent to the inequality in Eq. \ref{equation:IRpremise}.
\end{proof}

\newtheorem{proposition4}[propositionStart]{Proposition}
\begin{proposition4}[Condition for ex-post weak budget balancedness] \label{proposition:WBB}
    PVCG is ex-post weakly budget balanced (WBB) on the supply side i.f.f.
    the true capacity limits $\bar{\bm{x}}$, the true type parameters $\bm{\gamma}, \bm{\theta}$, and the adjustment payments $\bm{h}(\cdot)$ satisfy
    \begin{align}\label{equation:WBBpremise}
        \sum_{i=0}^{n-1} &h_i(\bar{\bm{x}}_{-i},\bm{\gamma}_{-i},\bm{\theta}) \le S^{*}(\bar{\bm{x}},\bm{\gamma},\bm{\theta})
        \nonumber \\
        & 
        -\sum_{i=0}^{n-1}[S^{*}(\bar{\bm{x}},\bm{\gamma},\bm{\theta}) -S^{*}_{-i}(\bar{\bm{x}}_{-i},\bm{\gamma}_{-i},\bm{\theta})].
    \end{align}
\end{proposition4}

\begin{proof}
    According to truthfulness proved in Proposition \ref{proposition:IC} and Theorem \ref{theorem:CremerMcLean},
    we use $\bar{\bm{x}}, \bm{\gamma}, \bm{\theta}$ to substitute $\hat{\bm{x}},\hat{\bm{\gamma}}, \hat{\bm{\theta}}$ in  Eq. \ref{equation:totalPayment} and
    Eq. \ref{equation:computeTau}.
    Then, the ex-post total payment to all producers is
    \begin{align}\label{equation:WBBproof}
        &\sum_{i=0}^{n-1}p_i(\bar{\bm{x}},\bm{\gamma},\bm{\theta}) = \sum_{i=0}^{n-1}[\tau_i(\bar{\bm{x}},\bm{\gamma},\bm{\theta})
        +h_i(\bar{\bm{x}}_{-i},\bm{\gamma}_{-i},\bm{\theta})]
        \nonumber \\
        &=\sum_{i=0}^{n-1}[S^{*}(\bar{\bm{x}},\bm{\gamma},\bm{\theta}) -S^{*}_{-i}(\bar{\bm{x}}_{-i},\bm{\gamma}_{-i},\bm{\theta})
        \nonumber \\
        &
        + c(\bar{x}_i \odot \eta_{i}^{*}(\bar{\bm{x}},\bm{\gamma},\bm{\theta}),\gamma_i)
        +h_i(\bar{\bm{x}}_{-i},\bm{\gamma}_{-i},\bm{\theta})].
    \end{align}
    As explained before, Theorem \ref{theorem:CremerMcLean} transforms the WBB condition to $\sum_{i=0}^{n-1}p_i(\bar{\bm{x}},\bm{\gamma},\bm{\theta}) \le \sum_{j=n}^{n+m-1}v(\bar{\bm{x}}\odot\bm{\eta}^{*}(\bar{\bm{x}},\bm{\gamma},\bm{\theta}),\theta_j)$.
    This inequality, together with the definition
    $S^{*}(\bar{\bm{x}},\bm{\gamma},\bm{\theta})=
    \sum_{j=n}^{n+m-1}v(\bar{\bm{x}}\odot\bm{\eta}^{*}(\bar{\bm{x}},\bm{\gamma},\bm{\theta}),\theta_j)-\sum_{i=0}^{n-1}c(\bar{x}_i
    \odot \eta_i^{*}(\bar{\bm{x}},\bm{\gamma},\bm{\theta}), \gamma_i)$,
    transforms Eq. \ref{equation:WBBproof} to Eq. \ref{equation:WBBpremise}.
\end{proof}

Then, we prove Eq. \ref{equation:calculateAdjustment} is a sufficient and necessary confition for Proposition \ref{proposition:IR} and Proposition \ref{proposition:WBB}.

\newtheorem{corollaryStart}{Corollary}
\begin{corollaryStart} \label{corollary:IRWBBequation}
    Under PVCG, a sufficient and necessary condition for ex-post IR and ex-post WBB to coexist on the supply side is
    \begin{equation}\label{equation:IRWBBtogetherLoss}
       \textrm{LOSS} = \textrm{Loss1} + \textrm{Loss2} = 0,
    \end{equation}
    where
    \begin{align} \label{equation:minimizeIR}
        & \textrm{Loss1} = \sum_{i=0}^{n-1} \textrm{ReLu}[
            - (S^{*}(\bar{\bm{x}},\bm{\gamma},\bm{\theta})
            -S^{*}_{-i}(\bar{\bm{x}}_{-i},\bm{\gamma}_{-i},\bm{\theta}))
        \nonumber \\
        & \quad \quad  -h_i(\bar{\bm{x}}_{-i},\bm{\gamma}_{-i}, \bm{\theta})
        ] \quad \textrm{and}
        \\
        \label{equation:minimizeWBB}
        & \textrm{Loss2} = \textrm{ReLu}
        [\sum_{i=0}^{n-1}[(S^{*}(\bar{\bm{x}},\bm{\gamma},\bm{\theta})
         -S^{*}_{-i}(\bar{\bm{x}}_{-i},\bm{\gamma}_{-i},\bm{\theta}))
        \nonumber \\
        & \quad \quad +h_i(\bar{\bm{x}}_{-i},\bm{\gamma}_{-i},\bm{\theta}) ]
        -S^{*}(\bar{\bm{x}},\bm{\gamma},\bm{\theta})].
    \end{align}
\end{corollaryStart}
\begin{proof}
    $Loss1 \ge 0, Loss2 \ge 0, \forall \bar{\bm{x}}, \bm{\gamma}, \bm{\theta}$. $Loss1 = 0$ i.f.f. Eq. \ref{equation:IRpremise} holds; $Loss2 = 0$ i.f.f. Eq. \ref{equation:WBBpremise} holds. 
\end{proof}

From Corollary \ref{corollary:IRWBBequation}, we know that if the adjustment payments $\bm{h}(\bar{\bm{x}}, \bm{\gamma}, \bm{\theta})$ are such chosen that Eq. \ref{equation:IRWBBtogetherLoss} holds for true capacity limits $\bar{\bm{x}}$ and true type parameters $\bm{\gamma}, \bm{\theta}$, then PVCG attains IR and WBB simultaneously. However, these true parameters are unknown to the coalition coordinator a priori, we need to find the \emph{optimal adjustment payment function} $\bm{h}^*(\cdot)$ such that for all possible $\bar{\bm{x}}$ and $\bm{\gamma}$ drawn from their respective prior distributions, Eq. \ref{equation:IRWBBtogetherLoss} holds. This is equivalent to solve the functional equation in Eq. \ref{equation:calculateAdjustment}. When such a functional solution exists, it minimizes the expected value of $LOSS=Loss1+Loss2$.

\newtheorem{corollaryMinimizeLoss}[corollaryStart]{Corollary}
\begin{corollaryMinimizeLoss} \label{corollary:minimizeLoss}
    The solution $\bm{h}^*(\cdot)$ to Eq. \ref{equation:calculateAdjustment}, if existing, is also a solution to the following minimization problem:
    \begin{equation}\label{equation:minimizeLoss}
        \bm{h}^*(\cdot) = \textrm{argmin}_{\bm{h}(\cdot)} 
            \mathbb{E}_{(\bar{\bm{x}},\bm{\gamma}, \bm{\theta})} 
            \{ \textrm{LOSS} \},
    \end{equation}
    where the expectation is over the prior distribution of $(\bar{\bm{x}},\bm{\gamma},\bm{\theta})$.
\end{corollaryMinimizeLoss}
\begin{proof}
    For the solution  $\bm{h}^*(\cdot)$ to Eq. \ref{equation:calculateAdjustment}, 
    \\ 
    $LOSS \equiv 0, \, \forall (\bm{x},\bm{\gamma}, \bm{\theta}) \in \textrm{supp}(\textrm{Prior}(\bm{x},\bm{\gamma}, \bm{\theta}))$. Therefore, $\mathbb{E}_{(\bar{\bm{x}},\bm{\gamma}, \bm{\theta} ) } \{ LOSS \} = 0 $. 
    
    For all other function $\bm{h}(\cdot)$, since $LOSS \ge 0, \, \forall (\bar{\bm{x}},\bm{\gamma}, \bm{\theta})$, we have $\mathbb{E}_{(\bar{\bm{x}},\bm{\gamma}, \bm{\theta})} \{ LOSS \} \ge 0$. 
\end{proof}

From Corollary \ref{corollary:minimizeLoss}, we can find the functional solution $\bm{h}^*(\cdot)$ to Eq. \ref{equation:calculateAdjustment} by minimizing the expected $LOSS$. The remaining problem is to prove the existence of such a solution. The following theorem holds.

\newtheorem{theorem1}[theoremstart]{Theorem}
\begin{theorem1} \label{theorem:WBBIRtogether}
    The following inequality is a sufficient and necessary condition for the existence of
    $\bm{h}^*(\cdot)$ such that the PVCG payments $\bm{p}(\cdot)
    = \bm{\tau}(\cdot) + \bm{h}^*(\cdot)$ satisfy WBB and IR
    for all $(\bar{\bm{x}},\bm{\gamma},\bm{\theta}) \in \textrm{supp}(\textrm{Prior}(\bar{\bm{x}},\bm{\gamma}, \bm{\theta}))$.
    \begin{align} \label{equation:IRWBB}
        &\sum_{i=0}^{n-1}[
            S^{*}(
                \bar{\bm{x}},
                \bm{\gamma},\bm{\theta}) - S^{*}(
                    (\min \bar{x}_i,\bm{x}_{-i}),
                    (\max \gamma_i,\bm{\gamma}_{-i})
        ,
            \bm{\theta}
                )
        ]
        \nonumber
                    \\
        &
         \le S^{*}(\bar{\bm{x}},\bm{\gamma},\bm{\theta}),
         \forall
         (\bar{\bm{x}},\bm{\gamma},\bm{\theta}) \in \textrm{supp}(\textrm{Prior}(\bar{\bm{x}},\bm{\gamma}, \bm{\theta})),
    \end{align}
    where $\min \bar{x}_i$ and $\max \gamma_i$ are the extreme values of $\bar{\bm{x}}$ and $\bm{\gamma}$ on their support set $\textrm{supp}(\textrm{Prior}(\bar{\bm{x}},\bm{\gamma}, \bm{\theta}))$.
\end{theorem1}

In order to prove Theorem \ref{theorem:WBBIRtogether}, we introduce the following lemma first.
\newtheorem{lemmaStart}{Lemma}
\begin{lemmaStart} \label{lemma:surplusIncreasing}
    The maximum social surplus monotonically increases with $\bar{x}_i$ and monotonically decreases with $\gamma_i$ for every producer $i$.
\end{lemmaStart}

\begin{proof}
    Suppose $\bar{\bm{x}}' \ge \bar{\bm{x}}$; then, by definition,
    \begin{equation} \label{equation:lemma1Proof1}
        S^{*}(\bar{\bm{x}}', \bm{\gamma},\bm{\theta}) \ge S(\bar{\bm{x}}' \odot \bm{\eta}, \bm{\gamma},\bm{\theta}), 
        \, \forall \bm{\eta} \in [0,1]^{\textrm{dim}(x_i) \times n}.
    \end{equation}
    Particularly, Eq. \ref{equation: propositionSSMproof2} holds for 
    $\bm{\eta} =\bm{\eta}^{*}(\bar{\bm{x}},\bm{\eta},\bm{\theta}) \odot \bar{\bm{x}} \oslash \bar{\bm{x}}' \in [0,1]^{\textrm{dim}(x_i) \times n}$, i.e.,
    \begin{align} \label{equation:lemma1Proof2}
        S^{*}(\bar{\bm{x}}', \bm{\gamma}, \bm{\theta})
        & \ge
        S(\bar{\bm{x}}' \odot
        \bm{\eta}^{*}(\bar{\bm{x}},\bm{\gamma},\bm{\theta}) \odot \bar{\bm{x}} \oslash \bar{\bm{x}}'
        , \bm{\gamma}, \bm{\theta})
        \nonumber \\
        & =
        S^{*}(\bar{\bm{x}}, \bm{\gamma}, \bm{\theta}).
    \end{align}
    Therefore, $S^{*}(\bar{\bm{x}},\bm{\gamma}, \bm{\theta})$ increases with $\bar{x}_i$.

    That $S^{*}(\bar{\bm{x}},\bm{\gamma},\bm{\theta})$ decreases with $\gamma_i$ is
    due to the \emph{envelope theorem}, which results in
    \begin{equation}\label{equation:envelop}
        \frac{\partial S^{*}}{\partial \gamma_i} |_{\bm{\eta}=\eta^{*}}
        =\frac{\partial S}{\partial \gamma_i} |_{\bm{\eta}=\eta^{*}}
        =-\frac{\partial c_i}{\partial \gamma_i} |_{\bm{\eta}=\eta^{*}} \le 0.
    \end{equation}
    %  \begin{equation}\label{equation:envelop}
    %    \textrm{and}\quad \frac{\partial S^{*}}{\partial \theta_i} |_{\bm{\eta}=\eta^{*}}
    %     =\frac{\partial S}{\partial \theta_i} |_{\bm{\eta}=\eta^{*}}
    %     =\frac{\partial v_i}{\partial \theta_i} |_{\bm{\eta}=\eta^{*}} \ge 0.
    % \end{equation}
\end{proof}

Now we can prove Theorem \ref{theorem:WBBIRtogether} as follows.
\begin{proof}[Proof of Theorem \ref{theorem:WBBIRtogether}]
    We prove sufficiency first. 
    
    If Eq. \ref{equation:IRWBB} holds, we aim to prove\\
    $h_i(\bar{\bm{x}}_{-i},\bm{\gamma}_{-i}, \bm{\theta})
    \equiv
    -[S^{*}((\min \bar{x}_i,\bar{\bm{x}}_{-i}),(\max \gamma_i,\bm{\gamma}_{-i}),\bm{\theta})
    -S^{*}_{-i}(\bar{\bm{x}}_{-i},\bm{\gamma}_{-i},\bm{\theta})], \, i=0,\ldots,n-1$
    satisfy Eq. \ref{equation:IRpremise} and \ref{equation:WBBpremise}
    for all $(\bar{\bm{x}},\bm{\gamma},\bm{\theta}) \in \textrm{supp}(\textrm{Prior}(\bar{\bm{x}},\bm{\gamma}, \bm{\theta}))$.
    
    In this case,
    all $(\bar{\bm{x}},\bm{\gamma},\bm{\theta})$ satisfy Eq. \ref{equation:IRpremise}
    because of the monotonic property
    proved in Lemma \ref{lemma:surplusIncreasing}, i.e.,
    \begin{align}
        & S^{*}((\min \bar{x}_i,\bar{\bm{x}}_{-i}),(\max \gamma_i,\bm{\gamma}_{-i}),\bm{\theta})
        \nonumber \\
        &
        \le
        S^{*}(\bar{\bm{x}},\bm{\gamma},\bm{\theta}),
        \, \forall \bar{\bm{x}},\bm{\gamma},\bm{\theta}
        \in \textrm{supp}(\textrm{Prior}(\bar{\bm{x}},\bm{\gamma}, \bm{\theta})).
    \end{align}
    
    Meanwhile, Eq. \ref{equation:WBBpremise} becomes
    \begin{align} \label{equation:IRWBBtogetherProof2}
        &
        \sum_{i=0}^{n-1}
        \{ -[S^{*}((\min \bar{x}_i,\bar{\bm{x}}_{-i}),(\max \gamma_i,\bm{\gamma}_{-i}),\bm{\theta})
        \nonumber \\
        &
        -S^{*}_{-i}(\bar{\bm{x}}_{-i},\bm{\gamma}_{-i},\bm{\theta})] \}
        \le S^{*}(\bar{\bm{x}},\bm{\gamma},\bm{\theta})
        \nonumber \\
        &
        -\sum_{i=0}^{n-1}[S^{*}(\bar{\bm{x}},\bm{\gamma},\bm{\theta})
        -S^{*}_{-i}(\bar{\bm{x}}_{-i},\bm{\gamma}_{-i}, \bm{\theta})],
    \end{align}
    which can be deduced from Eq. \ref{equation:IRWBB}.

    To prove necessity, we set $\bar{\bm{x}} = (\min \bar{x}_i, \bar{\bm{x}}_{-i})$ and
    $\bm{\gamma} = (\max \gamma_i, \bm{\gamma}_{-i})$ in 
    Eq. \ref{equation:IRpremise} and get a necessary condition: 
    \begin{align} \label{equation:IRWBBtogetherProof}
        &h_i(\bar{\bm{x}}_{-i},\bm{\gamma}_{-i},\bm{\theta})
        \ge -   
        [S^{*}((\min \bar{x}_i,\bar{\bm{x}}_{-i}),
        (\max \gamma_i,\bm{\gamma}_{-i})
       ,\bm{\theta}) 
        \nonumber
        \\
        &
         -S^{*}_{-i}(\bar{\bm{x}}_{-i},\bm{\gamma}_{-i},\bm{\theta})]
        ,  \forall \bar{\bm{x}},\bm{\gamma},\bm{\theta} \in \textrm{supp}(\textrm{Prior}(\bar{\bm{x}},\bm{\gamma}, \bm{\theta})).
    \end{align}
    Eq. \ref{equation:IRWBBtogetherProof}, together with Eq.
    \ref{equation:WBBpremise}, results in Eq. \ref{equation:IRWBB}.
\end{proof}

There are some useful corollaries of Theorem \ref{theorem:WBBIRtogether}.

\newtheorem{corollary1}[corollaryStart]{Corollary}
\begin{corollary1} \label{corollary:accurateEstimation}
    The functional solution $\bm{h}^*(\cdot)$ to Eq. \ref{equation:calculateAdjustment} exists if the prior estimation on
    $(\bar{\bm{x}},\bm{\gamma}$ is accurate enough.
\end{corollary1}
\begin{proof}
    When the the prior estimation on $(\bar{\bm{x}},\bm{\gamma})$ is accurate
    enough, $\bar{x}_i$ approaches to $\min \bar{x}_i$ and $\gamma_i$ approaches to $\max \gamma_i)$; hence, the left side of Eq. \ref{equation:IRWBB} approaches $0$, while
    the right side approaches a positive number.
\end{proof}

\newtheorem{corollary3}[corollaryStart]{Corollary}
\begin{corollary3} \label{corollary:superAdditivityLeadsToCorollary2}
    If the individual valuation function is super additive and with decreasing cross marginal returns, 
    there exists at least one solution $\bm{h}^*(\cdot)$ to Eq. \ref{equation:calculateAdjustment} such that dominant incentive compatibility, ex-post allocative efficiency, ex-post
    individual rationality,
    and ex-post weak budget balancedness are simultaneously attained on the supply side.
\end{corollary3}

\begin{proof}
% First, it is easy to see when the payoff function
% is super additive and with decreasing cross marginal returns,
% the social surplus function also has these two properties, i.e., 
% \begin{align}
%         \label{equation:superAdditivitySurplus}
%        &S(\bm{x}, \bm{\gamma},\bm{\theta})
%        \ge
%        \sum_{i=0}^{n-1}S(x_i, \gamma_i,\theta_i),
%        \, \forall \bm{x},\bm{\gamma}, \, \textrm{and}
%        \\
%        \label{equation:decreasingMarginalSurplus}
%        &S((x_i,\bm{x}_{-i}),\bm{\gamma},\bm{\theta})
%        -
%        S((x_i',\bm{x}_{-i}), \bm{\gamma},\bm{\theta})
%        \\
%        &\le
%        S((x_i,\bm{x}_{-i}'), \bm{\gamma},\bm{\theta})
%        -
%        S((x_i',\bm{x}_{-i}'), \bm{\gamma},\bm{\theta}),
%        \, \forall x_i \ge x_i', \bm{x}_{-i} \ge \bm{x}_{-i}'.
%        \nonumber
% \end{align}

According to Lemma \ref{lemma:surplusIncreasing},
\begin{align}
& S^{*}((\min\bar{x}_i, \bar{\bm{x}}_{-i}),(\max \gamma_i, \bm{\gamma}_{-i}), \bm{\theta} )
\\
& \ge S^{*}((\bm{0}, \bar{\bm{x}}_{-i}),(\max \gamma_i, \bm{\gamma}_{-i}), \bm{\theta} ) = S^{*}((\bm{0}, \bar{\bm{x}}_{-i}),\bm{\gamma}, \bm{\theta}).
\nonumber
\end{align}
Hence, a sufficient condition for Eq. \ref{equation:IRWBB} to hold in this case is
    \begin{align} \label{equation:IRWBBnoInformationAsymmetry}
        &\sum_{i=0}^{n-1}[
            S^{*}(
                \bar{\bm{x}},
                \bm{\gamma},\bm{\theta}) -
                S^{*}((\bm{0}, \bar{\bm{x}}_{-i}),\bm{\gamma}, \bm{\theta} )
        ]
        \nonumber \\
        &
         \le S^{*}(\bar{\bm{x}},\bm{\gamma},\bm{\theta}), \quad \forall \bar{\bm{x}},\bm{\gamma},\bm{\theta} \in \textrm{supp}(\textrm{Prior}(\bar{\bm{x}},\bm{\gamma}, \bm{\theta})).
    \end{align}

For arbitrary $(\bar{\bm{x}},\bm{\gamma}, \bm{\theta})\in \textrm{supp}(\textrm{Prior}(\bar{\bm{x}},\bm{\gamma},\bm{\theta}))$, suppose $\bm{\eta}^{*}$ maximizes
$S(\bar{\bm{x}} \odot \bm{\eta}, \bm{\gamma}, \bm{\theta})$. The corresponding
optimal social surplus are denoted by  $S^{*}(\bar{\bm{x}},\bm{\gamma},\bm{\theta})$, i.e.,
\begin{align}\label{equation:etaCorollary1}
    &S^{*}(\bar{\bm{x}},\bm{\gamma},\bm{\theta})
    = S(\bar{\bm{x}}\odot \bm{\eta}^{*}, \bm{\gamma}, \bm{\theta})
    \nonumber \\
    &
       = \sum_{j=n}^{n+m-1}v(\bar{\bm{x}}\odot\bm{\eta}^{*},\theta_j)-\sum_{i=0}^{n-1}c(\bar{x}_i \odot \eta^{*}_i, \gamma_i).
    % \\ \label{equation:etaCorollary2}
    % &    S_{-i}^{*}(\bar{\bm{x}}_{-i},\bm{\gamma}_{-i}, \bm{\theta}_{-i})
    % =
    %     \sum_{k \ne i} [v(\bar{\bm{x}}_{-i}\odot\bm{\eta}^{-i*},\theta_k)
    % \nonumber \\
    % & \qquad \qquad \qquad -c(\bar{x}_k \odot \eta^{-i*}_k, \gamma_k)],
    % \, \textrm{and}
    % \\ \label{equation:etaCorollary3}
    % & S^{*}_i(\bar{x}_i,\gamma_i)
    % =
    %     v(\bar{x}_i\odot \eta^{i*}, \theta_i)
    %     -c_i(\bar{x}_i \odot \eta^{i*}, \gamma_i).
\end{align}
We use $\eta^{*}_{i}$ and $\bm{\eta}^{*}_{-i}$ to denote the components of $\bm{\eta}^{*}$.
% Note that $\eta^{*}_i$ and $\eta^{i*}$ are different: the former is the $i$th component of $\bm{\eta}^{*}$, which maximizes $S(\bar{\bm{x}} \odot \bm{\eta},\bm{\gamma}, \bm{\theta})$, while the latter is the optimal $\eta_i$ that maximizes $S(\bar{x}_i \odot \eta_i, \gamma_i, \theta_i)$. Similarly, $\bm{\eta}^{*}_{-i}$ and $\bm{\eta}^{-i*}$ are different.

According to the definition of $S^{*}((\bm{0}, \bar{\bm{x}}_{-i}),\bm{\gamma}, \bm{\theta})$
, we know
\begin{align}
%   \label{equation:socialSurplusCompare1}
%   S(\bar{\bm{x}} \odot \bm{\eta}^{*},\bm{\gamma},\bm{\theta})
%   = & S^{*}(\bar{\bm{x}},\bm{\gamma},\bm{\theta}), \, \textrm{and}
%   \\
  \label{equation:socialSurplusCompare2}
  S^{*}((\bm{0},\bar{\bm{x}}_{-i}),\bm{\gamma}, \bm{\theta})
  \ge
  S((\bm{0},\bar{\bm{x}}_{-i}) \odot \bm{\eta},\bm{\gamma}, \bm{\theta})
  ,\, \forall \bm{\eta},
\end{align}
which is particularly true for $\bm{\eta} = \bm{\eta}^{*}$

According to Assumption \ref{assumption:decreasingMarginal} and \ref{assumption:zeroContribution},
\begin{align}\label{equation:marginalDecreaseProof1}
  & v(\bar{\bm{x}} \odot \bm{\eta}^{*},\theta_j)
  -
  v((\bm{0},\bar{\bm{x}}_{-i}) \odot \bm{\eta}^{*},\theta_j)
  \nonumber \\
  =
    & v((\bar{x}_i \odot \eta^{*}_i ,\bar{\bm{x}}_{-i} \odot \bm{\eta}^{*}_{-i}),\theta_j)
  -
  v((\bm{0},\bar{\bm{x}}_{-i} \odot \bm{\eta}^{*}_{-i}),\theta_j)
  \nonumber \\
  \le &
  v((\bar{x}_i \odot \eta^{*}_i , \bm{0},\ldots,\bm{0}),\theta_j)
  -
  v((\bm{0},\bm{0},\ldots,\bm{0}),\theta_j)
  \nonumber \\
  = &
  v(\bar{x}_i \odot \eta^{*}_i,\theta_j)
  , \, \forall i \in N, j\in M.
\end{align}

% $v((\bar{x}_k,\bm{0},\ldots, \bm{0}) \odot \bm{\eta},\underline{\theta_i})=v(\bar{x}_k \odot \eta_k,\underline{\theta_i}), \,\forall \bm{\eta}, i, k$ , $v((\bm{0},\bm{0},\ldots,\bm{0}),\theta_i)=0$,  and noticing that  
% $v((\bm{0},\bar{\bm{x}}_{-k}) \odot \bm{\eta},\underline{\theta_i})=v((\bm{0},\bar{\bm{x}}_{-k}) \odot \bm{\eta},\underline{\theta_i}), \,\forall \bm{\eta}, i, k$, Eq. \ref{equation:marginalDecreaseProof1} becomes
% \begin{align}\label{equation:marginalDecreaseProof3}
%   & v(\bar{\bm{x}} \odot \bm{\eta}^{*},\underline{\theta}_i)
%   -
%   v((\bm{0},\bar{\bm{x}}_{-k}) \odot \bm{\eta}^{*},\underline{\theta}_{i})
%   \nonumber \\
%   \le &
%   v(\bar{x}_k \odot \eta^{*}_k,\underline{\theta}_i)
%   -
%   v(\bm{0},\underline{\theta}_i)
%   = v(\bar{x}_k \odot \eta^{*}_k,\underline{\theta}_i)
%   .
% \end{align}

% Then, we have
% \begin{align}\label{equation:optimalSocialSurplusMarginalDecrease}
%   & S^{*}(\bar{\bm{x}},\bm{\gamma},\bm{\theta})
%   -
%   S_{-i}^{*}(\bar{\bm{x}}_{-i},\bm{\gamma}_{-i},\bm{\theta}_{-i}).
%   \nonumber \\
%   \le &
%   S(\bar{x}_i \odot \eta^{*}_i,\gamma_i,\theta_i)
%   -
%   S(\bm{0},\gamma_i,\theta_i).
% \end{align}

% According to super additivity and $S(\bm{0},\gamma_i)=0$, we derive
Then, we can derive
\begin{align}
   \label{equation:optimalSocialSurplusSuperAdditivity}
   & S^{*}(\bar{\bm{x}},\bm{\gamma},\bm{\theta})
    =
    S(\bar{\bm{x}}\odot\bm{\eta}^{*}, \bm{\gamma},\bm{\theta})
    %   & \ge \sum_{i=0}^{n-1}B(d_i \odot \eta^{*}_i)
%    -\sum_{i=0}^{n-1}c_i(d_i \odot \eta^{*}_i, \gamma_i)
%    \\
%   & =\sum_{i=0}^{n-1}[B(d_i \odot \eta^{*}_i)
%    - c_i(d_i \odot \eta^{*}_i, \gamma_i)]
%    \nonumber \\
  \nonumber \\
    &    
    = \sum_{j=n}^{n+m-1}v(\bar{\bm{x}} \odot \bm{\eta}^{*},\theta_j)
    -\sum_{i=0}^{n-1}c(\bar{x}_i \odot \eta^{*}_i,\gamma_i)
    \nonumber \\
    &    
    \ge \sum_{j=n}^{n+m-1}\sum_{i=0}^{n-1} v(\bar{x}_i \odot \eta^{*}_i,\theta_j)
    - \sum_{i=0}^{n-1} c(\bar{x}_i \odot \eta^{*}_i,\gamma_i)
    \nonumber \\
    &
     \ge \sum_{i=0}^{n-1}\{ \sum_{j=n}^{n+m-1}[ 
         v(\bar{\bm{x}} \odot \bm{\eta}^{*},\theta_j)
  -
  v((\bm{0},\bar{\bm{x}}_{-i}) \odot \bm{\eta}^{*},\theta_j)]
    \nonumber \\
    &
    \quad 
    - c(\bar{x}_i \odot \eta^{*}_i,\gamma_i)\}
    \nonumber \\
    &
       = \sum_{i=0}^{n-1} [\sum_{j=n}^{n+m-1} 
     v(\bar{\bm{x}} \odot \bm{\eta}^{*},\theta_j)
    -\sum_{k=0}^{n-1}c(\bar{x}_k \odot \eta^{*}_k,\gamma_k)]
    \nonumber \\
    &
  \quad - \sum_{i=0}^{n-1} [ \sum_{j=n}^{n+m-1}  v((\bm{0}, \bar{\bm{x}}_{-i}) \odot \bm{\eta}^{*},\theta_{j})
  \nonumber \\
  &
  \quad \quad \quad 
   - c(\bm{0} \odot \eta^{*}_i,\gamma_i)
  - \sum_{k=0, \ne i}^{n-1}
     c(\bar{x}_k \odot \eta^{*}_k,\gamma_k)]
    \nonumber \\
    &
    %    = \sum_{i=0}^{n-1} \sum_{k=0}^{n-1}[
    %  v(\bar{\bm{x}} \odot \bm{\eta}^{*},\underline{\theta_i})
    % - c(\bar{x}_k \odot \eta^{*}_k,\gamma_k)]
    % \nonumber \\
    % &
    % \quad
    % - \sum_{k=0}^{n-1} \sum_{i=0}^{n-1} v((\bm{0}, \bar{\bm{x}}_{-k})
    % \odot \bm{\eta}^{*},\underline{\theta_{i}})
    % \nonumber \\
    % &
    % \quad
    % -\sum_{i=0}^{n-1} \{
    % c(\bm{0} \odot \eta^{*}_i,\gamma_i)+
    % \sum_{k\ne i} c(\bar{x}_k \odot \eta^{*}_k,\gamma_k)]
    % \}
    % \nonumber \\
    % &
    =\sum_{i=0}^{n-1}[
        S(\bar{\bm{x}} \odot \bm{\eta}^{*},\bm{\gamma},\bm{\theta})
       - S((\bm{0}, \bar{\bm{x}}_{-i}) \odot \bm{\eta}^{*},\bm{\gamma}, \bm{\theta})]
    \nonumber \\
    &
    \ge \sum_{i=0}^{n-1}[
         S^{*}(\bar{\bm{x}},\bm{\gamma},\bm{\theta})
         - S^{*}((\bm{0}, \bar{\bm{x}}_{-i}),\bm{\gamma},\bm{\theta})
    ]
    ,
\end{align}
where the first ineqaulity is due to super additivity, the second inequality is a result from Eq. \ref{equation:marginalDecreaseProof1}, and the third inequality is from Eq. \ref{equation:etaCorollary1} and \ref{equation:socialSurplusCompare2}. Eq. \ref{equation:optimalSocialSurplusSuperAdditivity} is true for arbitrary $(\bar{\bm{x}},\bm{\gamma}, \bm{\theta})\in \textrm{supp}(\textrm{Prior}(\bar{\bm{x}},\bm{\gamma},\bm{\theta}))$, so
we get Eq. \ref{equation:IRWBBnoInformationAsymmetry}, which is a sufficient condition for the existence of $\bm{h}^*(\cdot)$ as a solution to \ref{equation:calculateAdjustment}.
\end{proof}

Lastly, it is worth pointing out even for those cases where there is no solution to Eq. \ref{equation:calculateAdjustment} (i.e., Assuption \ref{assumption:superAdditivity}-\ref{assumption:cremer} do not hold), an optimal $\bm{h}^{*}$ can still be found by minimizing the expected $LOSS$ in Eq. \ref{equation:minimizeLoss}. Such an optimal $\bm{h}^{*}$ guarantees the expected deviation from the IR and WBB constraints are minimized.

\section{Learning the Optimal Adjustment Payments}

\begin{figure}[t]
    \centering
    \includegraphics[width = 0.9\columnwidth]{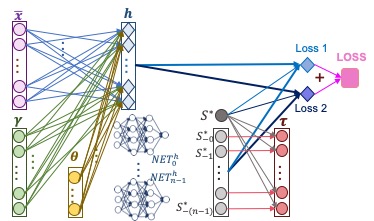}
    \caption{The structure of the composite neural network of PVCG}
    \label{fig:bigNetwork}
\end{figure}

% In general, there is no closed-form solution to Eq. \ref{equation:calculateAdjustment}. Nonetheless, 
Corollary \ref{corollary:minimizeLoss} informs us that we can learn the optimal adjustment payments $\bm{h}^*(\cdot)$ by minimizing the expected LOSS in Eq. \ref{equation:minimizeLoss}. Also, we know neural networks can approximate arbitrary continuous functions to arbitrary precisions \cite{funahashi1989approximate}. Therefore, we construct $n$ neural networks $\textrm{NET}^{h}_i, i \in N$ to approximate $h_i(\cdot), i \in N$. Output nodes of these $n$ networks, denoted by $\textrm{NET}^{h}_i.o, i \in N$, are combined into a single \emph{composite neural network} in Figure \ref{fig:bigNetwork} with the following loss function:

% The output nodes of $\textrm{NET}^{h}_i,\textrm{NET}^{g}_i, i \in N$, denoted by
% $\textrm{NET}^{h}_i.o,\textrm{NET}^{g}_i.o, i \in N$,
% are used to compute
% the following loss function of the composite neural network (of which the reasonableness will be proved in the next section):
\begin{align} \label{equation:lossFunction2}
&   \textrm{LOSS}
    = \frac{1}{T}\sum_{t=0}^{T-1}
    \{
    \sum_{i=0}^{n-1}
        \textrm{ReLu}[
            - (S^{*t} -S^{*t}_{-i})
          -
          \textrm{NET}^h_i.o
          ]
    \nonumber \\
    & 
    + \textrm{ReLu} [
    \sum_{i=0}^{n-1}(
    [(S^{*t}
     - S^{*t}_{-i})
    + \textrm{NET}^h_i.o]
        - S^{*t}
    )
    ],
\end{align}
where each sample $(\bar{\bm{x}}^t,\bm{\gamma}^t, \bm{\theta}^t)$ is drawn from their prior distribution $\textrm{Prior}(\bar{\bm{x}},\bm{\gamma},\bm{\theta})$ and $T$ is the sample size. For the $t$th sample, $\bm{\tau}^{t} = \bm{\tau}(\bar{\bm{x}}^{t},\bm{\gamma}^{t}, \bm{\theta}^t)$, $S^{*t}=S^{*}(\bar{\bm{x}}^{t},\bm{\gamma}^{t}, \bm{\theta}^t)$, and $S^{*t}_{-i}=S^{*}(\bar{\bm{x}}^{t}_{-i},\bm{\gamma}^{t}_{-i},\bm{\theta}^t_{-i})$.
Since we only need synthetic data to train this network, we can generate as many data as needed. As a result, LOSS can be minimized to its theoretical minimum almost perfectly in experiments.

\begin{figure}[t]
    \centering
    \includegraphics[width = 0.95\columnwidth]{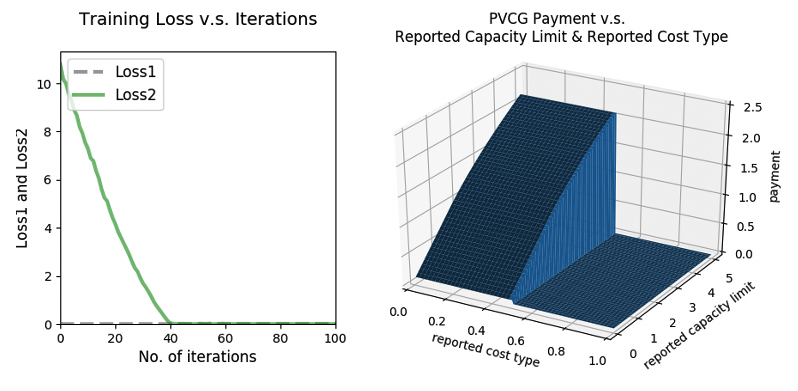}
    \caption{Training loss v.s. iterations (left) and PVCG payment v.s. reported capacity limit \& reported cost type (right)}
    \label{fig:paymentCurve}
\end{figure}
% begin{figure*}[t]
%    \centering
%    \includegraphics[width = 2.0\columnwidth]{figure_simulation.png}
%    \caption{Comparison between different FL payoff schemes}
%    \label{fig:simulationResult}
%\end{figure*}

To illustrate the effectiveness of this neural network method and to show the correctness of Corollary \ref{corollary:superAdditivityLeadsToCorollary2}, we learned the adjustment payments for a hypothetical scenario. We set
the individual valuation functions and individual cost functions as follows:
\begin{equation}\label{equation:revenueCostSpecification}
    v(\bm{x}) =\theta_i \sqrt{n (\sum_{k=0}^{n-1}{x}_k) }
    \, \textrm{and} \,
    c_i(x_i,\gamma_i) = \gamma_i x_i, i\in N.
\end{equation}
Note that the individual valuation function in Eq. \ref{equation:revenueCostSpecification} is a special case of that in Eq. \ref{eqaution:exampleValueFunction}, where $z(x_i)$ is set to be $z(x_i) = \sqrt{x_i}$. We report the experiment results for $n = 10, m=2$, $\textrm{Prior}(\bar{x}_i) = \textrm{Uniform}[0,5], i \in N$, $\textrm{Prior}(\gamma_i) = \textrm{Uniform}[0,1], i \in N$, and $\textrm{Prior}(\theta_j)=[0,1],j\in M$. We let
$\textrm{NET}^{h}_i, i\in N$ each have three $10$-dimensional hidden layers.
% and share parameters, reflecting the symmetry between participants a priori.

% \begin{figure}[t]
%     \centering
%     \includegraphics[width = 0.95\columnwidth]{fvcg_training_loss.png}
%     \caption{Training loss v.s. iterations}
%     \label{fig:lossCurve}
% \end{figure}

The loss curve is shown in the left figure of Fig.
\ref{fig:paymentCurve}. The training loss fast converges to $0$ after $40$ iterations, as expected. 
After we obtain the trained networks
$[\textrm{NET}^{h}_i], i \in N$, we can use trained networks to calculate PVCG payments $\bm{p}(\hat{\bm{x}},\hat{\bm{\gamma}}, \hat{\bm{\theta}})$ for any reported $(\hat{\bm{x}},\hat{\bm{\gamma}},\hat{\bm{\theta}})$. For illustration, we draw $p_0$, the
payment to participant $0$, with respect to $\hat{x}_0$ and $\hat{\gamma}_0$ in the right figure in Fig.
\ref{fig:paymentCurve}, fixing parameters of other participants at $\hat{x}_i \equiv
2.5, \hat{\gamma}_i \equiv 0.5, i \in N,\neq 0,  \hat{\theta}_j \equiv 0.5, j\in M$.

We can see that $p_0$ increases with
$\bar{x}_0$. This aligns well with our expectation that the higher the capacity limit a producer reports, the more input resources are contributed by this producer; hence, it receives higher payments. Also, $p_0$ remains constant with $\gamma_0$ when $\gamma_0$ is below a threshold and
sharply drops to around $0$ when $\gamma_0$ passes the threshold.
This implies that the payment to a producer should only be affected by its contribution to the cooperative production process rather than its cost, but if a producer's cost is too high, the optimal social choice is to exclude this producer from the production coalition and thus pay it nothing.

\section{Conclusions and Future Work}
In this paper, we presented the Procurement-VCG (PVCG) mechanism for incentivizing producers to collaboratively produce some valuable virtual products. PVCG incentivizes capacitated producers to truthfully report their capacity limits and their cost types, based on which we optimize for social surplus and guarantee ex-post allocative efficiency. With some reasonable assumptions, we prove that PVCG is ex-post individually rational and ex-post weakly budget balanced. To our knowledge, PVCG is among the very few mechanisms in the mechanism design literature that attain these four objectives simultaneously.

Since this paper focuses on theoretical analyses, we reported limited experimental results. When the number of producers increase, we require more advanced algorithms to learn the PVCG payments efficiently. Besides, for nonanalytic individual valuation functions and individual cost functions, calculating the optimal social surplus is computationally expensive as well. These computational issues are beyond the scope of this paper and left for future research. Also, in order to apply PVCG to scenarios such as federated learning, some special-purpose engineering designs (e.g., a sandbox) are required to measure the required parameters in our model, e.g., the input datasets and the output model. These implementation details will be discussed in our subsequent research. Moreover, in our following work, we will report simulation results that compare PVCG with other sharing rules for cooperative games, e.g., Shapley \cite{jia2019towards}, Labour Union \cite{gollapudi2017profit}, etc.

%\fontsize{9.5pt}{10.5pt}
%\selectfont
\bibliography{Bibliography-File}

\begin{thebibliography}{}

\bibitem[\protect\citeauthoryear{Albert, Conitzer, and
  Lopomo}{2015}]{albert2015assessing}
Albert, M.; Conitzer, V.; and Lopomo, G.
\newblock 2015.
\newblock Assessing the robustness of cremer-mclean with automated mechanism
  design.
\newblock In {\em Twenty-Ninth AAAI Conference on Artificial Intelligence}.

\bibitem[\protect\citeauthoryear{Bachrach \bgroup et al\mbox.\egroup
  }{2009}]{bachrach2009cost}
Bachrach, Y.; Meir, R.; Zuckerman, M.; Rothe, J.; and Rosenschein, J.~S.
\newblock 2009.
\newblock The cost of stability in weighted voting games.
\newblock In {\em Proceedings of The 8th International Conference on Autonomous
  Agents and Multiagent Systems-Volume 2},  1289--1290.
\newblock International Foundation for Autonomous Agents and Multiagent
  Systems.

\bibitem[\protect\citeauthoryear{Beberg \bgroup et al\mbox.\egroup
  }{2009}]{Beberg2009Folding}
Beberg, A.~L.; Ensign, D.~L.; Jayachandran, G.; Khaliq, S.; and Pande, V.~S.
\newblock 2009.
\newblock Folding@home: Lessons from eight years of volunteer distributed
  computing.
\newblock In {\em IEEE International Symposium on Parallel \& Distributed
  Processing}.

\bibitem[\protect\citeauthoryear{Branzei, Dimitrov, and
  Tijs}{2008}]{branzei2008models}
Branzei, R.; Dimitrov, D.; and Tijs, S.
\newblock 2008.
\newblock {\em Models in cooperative game theory}, volume 556.
\newblock Springer Science \& Business Media.

\bibitem[\protect\citeauthoryear{Cai, Daskalakis, and
  Papadimitriou}{2015}]{cai2015optimum}
Cai, Y.; Daskalakis, C.; and Papadimitriou, C.
\newblock 2015.
\newblock Optimum statistical estimation with strategic data sources.
\newblock In {\em Conference on Learning Theory},  280--296.

\bibitem[\protect\citeauthoryear{Chalkiadakis \bgroup et al\mbox.\egroup
  }{2010}]{chalkiadakis2010cooperative}
Chalkiadakis, G.; Elkind, E.; Markakis, E.; Polukarov, M.; and Jennings, N.~R.
\newblock 2010.
\newblock Cooperative games with overlapping coalitions.
\newblock {\em Journal of Artificial Intelligence Research} 39:179--216.

\bibitem[\protect\citeauthoryear{Chalkiadakis, Elkind, and
  Wooldridge}{2011}]{chalkiadakis2011computational}
Chalkiadakis, G.; Elkind, E.; and Wooldridge, M.
\newblock 2011.
\newblock Computational aspects of cooperative game theory.
\newblock {\em Synthesis Lectures on Artificial Intelligence and Machine
  Learning} 5(6):1--168.

\bibitem[\protect\citeauthoryear{Chandrashekar \bgroup et al\mbox.\egroup
  }{2007}]{chandrashekar2007auction}
Chandrashekar, T.~S.; Narahari, Y.; Rosa, C.~H.; Kulkarni, D.~M.; Tew, J.~D.;
  and Dayama, P.
\newblock 2007.
\newblock Auction-based mechanisms for electronic procurement.
\newblock {\em IEEE Transactions on Automation Science and Engineering}
  4(3):297--321.

\bibitem[\protect\citeauthoryear{Chen \bgroup et al\mbox.\egroup
  }{2005}]{chen2005efficient}
Chen, R.~R.; Roundy, R.~O.; Zhang, R.~Q.; and Janakiraman, G.
\newblock 2005.
\newblock Efficient auction mechanisms for supply chain procurement.
\newblock {\em Management Science} 51(3):467--482.

\bibitem[\protect\citeauthoryear{Christoforou \bgroup et al\mbox.\egroup
  }{2013}]{christoforou2013crowd}
Christoforou, E.; Anta, A.~F.; Georgiou, C.; Mosteiro, M.~A.; and S{\'a}nchez,
  A.
\newblock 2013.
\newblock Crowd computing as a cooperation problem: an evolutionary approach.
\newblock {\em Journal of Statistical Physics} 151(3-4):654--672.

\bibitem[\protect\citeauthoryear{Conitzer and
  Sandholm}{2006}]{conitzer2006complexity}
Conitzer, V., and Sandholm, T.
\newblock 2006.
\newblock Complexity of constructing solutions in the core based on synergies
  among coalitions.
\newblock {\em Artificial Intelligence} 170(6-7):607--619.

\bibitem[\protect\citeauthoryear{Cr{\'e}mer and
  McLean}{1985}]{creemer1985optimal}
Cr{\'e}mer, J., and McLean, R.~P.
\newblock 1985.
\newblock Optimal selling strategies under uncertainty for a discriminating
  monopolist when demands are interdepen0 denty.
\newblock {\em Econometrica} 53:345--361.

\bibitem[\protect\citeauthoryear{Drechsel and
  Kimms}{2010}]{drechsel2010computing}
Drechsel, J., and Kimms, A.
\newblock 2010.
\newblock Computing core allocations in cooperative games with an application
  to cooperative procurement.
\newblock {\em International Journal of Production Economics} 128(1):310--321.

\bibitem[\protect\citeauthoryear{Elkind \bgroup et al\mbox.\egroup
  }{2009}]{elkind2009computational}
Elkind, E.; Goldberg, L.~A.; Goldberg, P.~W.; and Wooldridge, M.
\newblock 2009.
\newblock On the computational complexity of weighted voting games.
\newblock {\em Annals of Mathematics and Artificial Intelligence}
  56(2):109--131.

\bibitem[\protect\citeauthoryear{Funahashi}{1989}]{funahashi1989approximate}
Funahashi, K.-I.
\newblock 1989.
\newblock On the approximate realization of continuous mappings by neural
  networks.
\newblock {\em Neural networks} 2(3):183--192.

\bibitem[\protect\citeauthoryear{Gollapudi \bgroup et al\mbox.\egroup
  }{2017}]{gollapudi2017profit}
Gollapudi, S.; Kollias, K.; Panigrahi, D.; and Pliatsika, V.
\newblock 2017.
\newblock Profit sharing and efficiency in utility games.
\newblock In {\em Proceedings of the 25th Annual European Symposium on
  Algorithms (ESA'17)}.

\bibitem[\protect\citeauthoryear{Iyengar and Kumar}{2008}]{iyengar2008optimal}
Iyengar, G., and Kumar, A.
\newblock 2008.
\newblock Optimal procurement mechanisms for divisible goods with capacitated
  suppliers.
\newblock {\em Review of Economic Design} 12(2):129.

\bibitem[\protect\citeauthoryear{Jackson}{2014}]{jackson2014mechanism}
Jackson, M.~O.
\newblock 2014.
\newblock Mechanism theory.
\newblock {\em Available at SSRN 2542983}.

\bibitem[\protect\citeauthoryear{Jia \bgroup et al\mbox.\egroup
  }{2019}]{jia2019towards}
Jia, R.; Dao, D.; Wang, B.; Hubis, F.~A.; Hynes, N.; Gurel, N.~M.; Li, B.;
  Zhang, C.; Song, D.; and Spanos, C.
\newblock 2019.
\newblock Towards efficient data valuation based on the shapley value.
\newblock {\em CoRR, arXiv:1902.10275}.

\bibitem[\protect\citeauthoryear{Kosenok and
  Severinov}{2008}]{kosenok2008individually}
Kosenok, G., and Severinov, S.
\newblock 2008.
\newblock Individually rational, budget-balanced mechanisms and allocation of
  surplus.
\newblock {\em Journal of Economic Theory} 140(1):126--161.

\bibitem[\protect\citeauthoryear{Kronbak and
  Lindroos}{2007}]{kronbak2007sharing}
Kronbak, L.~G., and Lindroos, M.
\newblock 2007.
\newblock Sharing rules and stability in coalition games with externalities.
\newblock {\em Marine Resource Economics} 22(2):137--154.

\bibitem[\protect\citeauthoryear{Larson \bgroup et al\mbox.\egroup
  }{2009}]{larson2009Folding}
Larson, S.~M.; Snow, C.~D.; Shirts, M.; and Pande, V.~S.
\newblock 2009.
\newblock Folding@ home and genome@ home: Using distributed computing to tackle
  previously intractable problems in computational biology.
\newblock {\em arXiv preprint arXiv:0901.0866}.

\bibitem[\protect\citeauthoryear{Niyato, Vasilakos, and
  Kun}{2011}]{niyato2011resource}
Niyato, D.; Vasilakos, A.~V.; and Kun, Z.
\newblock 2011.
\newblock Resource and revenue sharing with coalition formation of cloud
  providers: Game theoretic approach.
\newblock In {\em 2011 11th IEEE/ACM International Symposium on Cluster, Cloud
  and Grid Computing},  215--224.
\newblock IEEE.

\bibitem[\protect\citeauthoryear{Peleg and
  Sudh{\"o}lter}{2007}]{peleg2007introduction}
Peleg, B., and Sudh{\"o}lter, P.
\newblock 2007.
\newblock {\em Introduction to the theory of cooperative games}, volume~34.
\newblock Springer Science \& Business Media.

\bibitem[\protect\citeauthoryear{Rassias}{2012}]{rassias2012functional}
Rassias, T.
\newblock 2012.
\newblock {\em Functional equations and inequalities}, volume 518.
\newblock Springer Science \& Business Media.

\bibitem[\protect\citeauthoryear{Richardson, Filos-Ratsikas, and
  Faltings}{2019}]{richardson2019rewarding}
Richardson, A.; Filos-Ratsikas, A.; and Faltings, B.
\newblock 2019.
\newblock Rewarding high-quality data via influence functions.
\newblock {\em CoRR, arXiv:1908.11598}.

\bibitem[\protect\citeauthoryear{Sandholm and
  Lesser}{1997}]{sandholm1997coalitions}
Sandholm, T.~W., and Lesser, V.~R.
\newblock 1997.
\newblock Coalitions among computationally bounded agents.
\newblock {\em Artificial intelligence} 94(1):99--138.

\bibitem[\protect\citeauthoryear{Shen, Tang, and Zuo}{2019}]{shen2019automated}
Shen, W.; Tang, P.; and Zuo, S.
\newblock 2019.
\newblock Automated mechanism design via neural networks.
\newblock In {\em Proceedings of the 18th International Conference on
  Autonomous Agents and Multiagent Systems},  215--223.
\newblock International Foundation for Autonomous Agents and Multiagent
  Systems.

\bibitem[\protect\citeauthoryear{Wang}{2019}]{wang2019}
Wang, G.
\newblock 2019.
\newblock Interpret federated learning with shapley values.
\newblock {\em the 1st International Workshop on Federated Machine Learning for
  User Privacy and Data Confidentiality}.

\bibitem[\protect\citeauthoryear{Weikard}{2009}]{weikard2009cartel}
Weikard, H.-P.
\newblock 2009.
\newblock Cartel stability under an optimal sharing rule.
\newblock {\em The manchester school} 77(5):575--593.

\bibitem[\protect\citeauthoryear{Westenbroek \bgroup et al\mbox.\egroup
  }{2019}]{westenbroek2019competitive}
Westenbroek, T.; Dong, R.; Ratliff, L.~J.; and Sastry, S.~S.
\newblock 2019.
\newblock Competitive statistical estimation with strategic data sources.
\newblock {\em CoRR, arXiv:1904.12768}.

\bibitem[\protect\citeauthoryear{Yang \bgroup et al\mbox.\egroup
  }{2019a}]{yang2019federated}
Yang, Q.; Liu, Y.; Chen, T.; and Tong, Y.
\newblock 2019a.
\newblock Federated machine learning: Concept and applications.
\newblock {\em ACM Transactions on Intelligent Systems and Technology (TIST)}
  10(2):12.

\bibitem[\protect\citeauthoryear{Yang \bgroup et al\mbox.\egroup
  }{2019b}]{FL:2019}
Yang, Q.; Liu, Y.; Cheng, Y.; Kang, Y.; Chen, T.; and Yu, H.
\newblock 2019b.
\newblock {\em Federated Learning}.
\newblock Morgan \& Claypool Publishers.

\bibitem[\protect\citeauthoryear{Yu \bgroup et al\mbox.\egroup }{2020}]{yu2020}
Yu, H.; Liu, Z.; Liu, Y.; Chen, T.; Cong, M.; Weng, X.; Niyato, D.; and Qiang,
  Y.
\newblock 2020.
\newblock A fairness-aware incentive scheme for federated learning.
\newblock In {\em Proceedings of the 3rd AAAI/ACM Conference on Artificial
  Intelligence, Ethics, and Society (AIES-20)},  393--399.

\end{thebibliography}
\bibliographystyle{aaai}

\end{document}